\documentclass[12pt]{extarticle}
\usepackage[margin=0.3in, top=1in, bottom=1in]{geometry}
\usepackage[utf8]{inputenc}
\usepackage{amsmath,amssymb,amsthm,mathtools}
\usepackage{thmtools, thm-restate}
\usepackage{changepage}
\usepackage[T1]{fontenc}
\usepackage{eulervm,eucal,eufrak}% the math fonts used in concrete mathematics
\usepackage{bbold}
\usepackage{concrete}% the concrete-roman, used in concrete mathematics
\usepackage{lipsum,hyperref,datetime2}
\usepackage[dvipsnames]{xcolor}
\hypersetup{
    colorlinks=true,
    filecolor=red,
    linkcolor=blue,
    citecolor=blue,
    urlcolor=blue,
}
\usepackage{graphicx}
\graphicspath{ {./graphics/} }
\usepackage[backend=biber,style=alphabetic]{biblatex} % Citations: biber + biblatex
\bibliography{reference.bib} % pdflatex name.tex -> biber name -> pdflatex name.tex
\newtheoremstyle{mystyle}%                        % Name
  {}%                                             % Space above
  {}%                                             % Space below
  {}%                                             % Body font
  {}%                                             % Indent amount
  {\bfseries}%                                    % Theorem head font
  {}%                                             % Punctuation after theorem head
  { }%                                            % Space after theorem head, ' ', or \newline
  {\thmname{#1}\thmnumber{ #2}\thmnote{ (#3)}}%   % Theorem head spec (can be left empty, meaning `normal')
\numberwithin{equation}{section} % Kirillov thinks this is more usual.
\theoremstyle{mystyle}
\newtheorem{theorem}[equation]{Theorem}
\AtEndEnvironment{theorem}{\null\hfill$\diamond$}%

\AtEndEnvironment{convention}{\null\hfill$\diamond$}%
\newtheorem{proposition}[equation]{Proposition}
\AtEndEnvironment{proposition}{\null\hfill$\diamond$}%
\newtheorem{fact}[equation]{Fact}
\AtEndEnvironment{fact}{\null\hfill$\diamond$}%
\newtheorem{definition}[equation]{Definition}
\AtEndEnvironment{definition}{\null\hfill$\diamond$}%

\AtEndEnvironment{lemma}{\null\hfill$\diamond$}%

\AtEndEnvironment{corollary}{\null\hfill$\diamond$}%
\newtheorem{example}[equation]{Example}
\AtEndEnvironment{example}{\null\hfill$\diamond$}%
\newtheorem{remark}[equation]{Remark}
\AtEndEnvironment{remark}{\null\hfill$\diamond$}%

\AtEndEnvironment{notation}{\null\hfill$\diamond$}%
\let\oldproof\proof
\definecolor{my-dark-gray}{gray}{0.13}
\renewcommand{\proof}{\color{my-dark-gray}\oldproof}

\newtoggle{details}
\toggletrue{details}
\iftoggle{details}{
  \newcommand{\details}[1]{
    {\color{OliveGreen} #1} \\}
}{\newcommand{\details}[1]{}}

\usepackage{tikz}
\usetikzlibrary{shapes.geometric, arrows}
\usetikzlibrary{shapes.multipart} %% for multiple lines in node (must have align=center)
\usetikzlibrary{scopes}
\usetikzlibrary{math}
\usetikzlibrary{decorations.markings,decorations.pathreplacing}
\usetikzlibrary{fadings}
\usepackage{tikz-cd}

\tikzset{
every picture/.style={line width=0.8pt, >=stealth,
                       baseline=-3pt,label distance=-3pt},
emptynode/.style={circle,minimum size=0pt, inner sep=0pt, outer
sep=0},
dotnode/.style={fill=black,circle,minimum size=2.5pt, inner sep=1pt, outer
sep=0},
small_dotnode/.style={fill=black,circle,minimum size=2pt, inner sep=0pt, outer
sep=0},
morphism/.style={fill=white,circle,draw,thin, inner sep=1pt, minimum size=15pt,
                 scale=0.8},
small_morphism/.style={fill=white,circle,draw,thin,inner sep=1pt,
                       minimum size=10pt, scale=0.8},
ellipse_morphism/.style args={#1}{fill=white,circle,draw,thin,inner sep=1pt,
                       minimum size=5pt, scale=0.8,
												ellipse, draw, rotate=#1},
coupon/.style={draw,thin, inner sep=1pt, minimum size=18pt,scale=0.8},
semi_morphism/.style args={#1,#2}{
                  fill=white,semicircle,draw,thin, inner sep=1pt, scale=0.8,
                  shape border rotate=#1,
                  label={#1-90:#2}},
regular/.style={densely dashed}, %% for the regular color, i.e. sum d_i
edge/.style={very thick, draw=green, text=black},
overline/.style={preaction={draw,line width=2mm,white,-}},
thin_overline/.style={preaction={draw,line width=#1 mm,white,-}},
thin_overline/.default=2,
thick_overline/.style={preaction={draw,line width=3mm,white,-}},
really_thick/.style={line width=3mm, gray},
boundary/.style={thick,  draw=blue, text=black},
ribbon/.style={line width=1.5mm, postaction={draw,line width=1mm,white}},
ribbon_u/.style args={#1,#2}{line width=#1mm, postaction={draw,line width=#2mm,white}},
cell/.style={fill=black!10},
subgraph/.style={fill=black!30},
midarrow/.style={postaction={decorate},
                 decoration={
                    markings,% switch on markings
                    mark=at position #1 with {\arrow{>}},
                 }},
midarrow/.default=0.5,
midarrow_rev/.style={postaction={decorate},
                 decoration={
                    markings,% switch on markings
                    mark=at position #1 with {\arrow{<}},
                 }},
midarrow_rev/.default=0.5,
block/.style={rectangle, rounded corners, text centered, draw=black, align=center}
}

\tikzstyle{block} = [rectangle, rounded corners, text centered, draw=black, align=center]

\tikzfading[name=fade inside, inner color=transparent!80, outer
color=transparent!10]

\usepackage{svg}

\usetikzlibrary{arrows}
\usetikzlibrary{arrows.meta}
\usetikzlibrary{intersections}
\usetikzlibrary{decorations.markings}
\usetikzlibrary{decorations}
\usetikzlibrary{patterns}
\usepackage{tikz-3dplot}
\tikzset{yxplane/.style={canvas is xy plane at z=#1}}
\tikzset{>=latex}
\tikzset{->-/.style={decoration={
  markings,
  mark=at position .5 with {\arrow{>}}},postaction={decorate}}}
\tikzset{-<-/.style={decoration={
  markings,
  mark=at position .5 with {\arrow{<}}},postaction={decorate}}}
\usepackage{amscd}
\usepackage{psfrag}
\title{Equivalence of field theories: Crane-Yetter and the shadow}
\author{Jin-Cheng Guu}
\date{}

\begin{document}

\maketitle
\begin{flushright}
  \tiny{Compiled Time: [\today\,\DTMcurrenttime]} \quad\qquad.
\end{flushright}

\abstract{

  It has been open for years to clarify the relationship between
  two smooth $4$-manifolds invariants, the shadow model
  (motivated by statistical mechanics
  \cite{turaev/topology-of-shadows}) and the simplicial
  Crane-Yetter model (motivated by topological quantum field
  theory
  \cite{crane-yetter/a-categorical-construction-of-4d-tqft}),
  both of which degenerate to the $3$D Witten-Reshetikhin-Turaev
  model in a special case. Despite the seeming difference in
  their origins and formal constructions, we show that they are
  in fact equal.

  Along the way, we sketch a dictionary between the shadow model
  and the Crane-Yetter model, provide a brief survey to the
  shadow construction a la Turaev, and suggest once again that
  the semisimple models have reached their limits.

}
\tableofcontents

\section*{Acknowledgement}

The author would like to thank fruitful discussions with Oleg
Viro, Alexander Kirillov, Vladimir Turaev, Shamuel Auyeung, and
Jiahao Hu. The main result of this paper was conjectured by O.
Viro two decades ago; to paraphrase, ``We have two combinatorial
invariants for 4-manifolds. It would be a miracle if they are
different.'' V. Turaev, as a leading expert of this field,
confirmed that it was still open in mid $2021$. Without their
sharing, the author would not have worked on this problem.

\section{Introduction}

Topology is the wildest in dimension $4$. For example, the smooth
Poincare conjecture remains far from proven only for $n = 4$, and
the topological $\mathbb{R}^{n}$ admits exactly one
diffeomorphism type unless $n=4$, in which case uncountably many
are available. There are gauge-theoretic tools which, to some
extent, are sensitive to exotic smooth phenomena, such as the
Donaldson and Seiberg-Witten invariants. Despite their successes,
they are unable to tackle a large class of problems including the
smooth Poincare conjecture for $n=4$.

In the $90$s, a simpler invariant of smooth $4$-manifolds was
proposed by Crane and Yetter (CY). The original CY invariant
could only detect homotopy type but its simplicity leaves room
for modifications. Despite several attempts at modification (e.g.
\cite{barenz/evaluation-crane-yetter}), to date, there has not been
much success at detecting exotic smooth phenomena. A recent work
by Reutter \cite{reutter/semisimple} explains the failure, and
suggests the need for a non-semisimple or derived variant of the
CY model.

Before moving into that direction, the author aims to settle
another issue first. There is another invariant of
$4$-dimensional smooth manifolds, the shadow model a la V. Turaev
\cite{turaev/topology-of-shadows} \cite{turaev-qiok-3-manifolds},
from statistical mechanics. Moreover, it was known that the
shadow model coincides with the CY model when both degenerate
\cite[X.3.2 \& theorem X.3.3]{turaev-qiok-3-manifolds}
\cite{barrett/observables-in-tv-and-cy} to the $3$D
Witten-Reshetikhin-Turaev model (also known as the quantized
Chern-Simons theory). It is thus necessary to clarify their
relationship in the general semisimple case. Despite the
difference of their origins and formal definitions, this paper
shows them equal, suggesting once again that semisimple models
have reached their limit in terms of detecting exotic smooth
phenomena.

Along proving the equivalence of the two models, we make heavy
use of the construction of the shadow model given in
\cite{turaev-qiok-3-manifolds}. We include the essential details
of the construction in this paper which serve as a digestible
survey of the shadow model.

\subsection{Sections summary}

\begin{itemize}
  \item Section \ref{section/algebra}: We provide the basics of tensor
        categories, their graphical calculi and tensor networks,
        and special numerical entities ($n$j-symbols). Nothing in
        this section is new.
  \item Section \ref{section/topology}: We provide three kinds of
        data that present smooth $4$-manifolds: triangulations,
        handle decompositions, and shadows. We also fully recall
        the definition of a shadow, which closely resembles foams
        in the modern literature on Khovanov homology. Nothing in
        this section is new. Readers can treat this paper as a
        thin interface to the book
        \cite{turaev-qiok-3-manifolds}.
  \item Section \ref{section/sum}: We provide the definitions of
        the two state sums: the CY model and the shadow model.
        The novel observation is that the shadow model can be
        extended from modular categories to premodular
        categories. We conclude the section by stating and
        proving the equivalence.
\end{itemize}

\subsection{Note on the arXiv version}
The tex source file of this paper includes hidden details, which
can be displayed by recompiling with toggling \texttt{details} in
the source.

\subsection{A summary to experts}\label{subsection/a-summary-to-experts}

The idea of the proof for the equivalence is simple. Let $X$ be a
$4$-manifold. While the CY state sum can be computed from any
triangulation $T$ of $X$, the shadow state sum can be computed
from any stable shadow of $X$. We construct a natural stable
shadow $S$ from $T$ and compute the shadow sum in terms of $S$.
If the shadow sum is actually a $(3+1)$D-TQFT, we can reduce the
task to proving that the local shadow evaluates to the local term
involved in the CY state sum (namely, the $10$j-symbols).
However, the author could not prove the semi-locality, but rather
found a workaround for the case of closed $4$-manifolds. The
author expects the shadow sum can be modified to be a fully
extended TQFT, which is in turn fully equivalent to the
Crane-Yetter model.

\subsection{Conventions}

\noindent Some conventions we use globally in the paper:

\begin{itemize}
  \item We fix a algebraic closed field $\mathbb{k}$ of
        characteristic $0$.
  \item By a vector space $V$ we mean a finite dimensional vector
        space over the field $\mathbb{k}$, unless further
        specified. The linear dual
        $Hom_{\mathbb{k}}(V,\mathbb{k})$ is denoted by
        $V^{\star}$.
  \item By a manifold we mean a piecewise-linear, oriented,
        connected and closed manifold in real dimension $4$
        unless further specified.
  \item In this paper, by a monoidal category we mean a strict
        monoidal category (we do not lose any generality by Mac
        Lane's strictness theorem). \details{For the statement
        and a full modern proof, see \cite[theorem
        2.8.5]{egno/tensor-cats}.}
\end{itemize}

\section{Algebra $(A)$}\label{section/algebra}
\subsection{Premodular category}

The full definition of a premodular category from scratch is
tedious. Unfamiliar readers can think of a premodular category
roughly as a higher version of the group algebra of a finite
group. A formal definition can be found after some motivations
(\ref{def/premodular-category}).

Algebraic objects help abstract details in various mathematical
problems. However, they sometimes abstract too much to recover
information of interest. In recent years, mathematicians
``categorify'' algebraic objects in order to retain more
information. For example, a ring is categorified to a tensor
category, and a tensor category with special properties and
additional structures could be powerful. For example, a ribbon
fusion category provide quantum invariants of knots that
generalize the Jones polynomials. A premodular category is a
braided fusion category satisfying with a spherical structure.
Examples include the (modified) category of representations of
finite groups, finite $2$-groups, and quantum groups.

\begin{definition}[premodular category]\label{def/premodular-category}
  A premodular category is a spherical braided fusion category.
\end{definition}

\noindent In particular, a premodular category $C$ is semisimple,
$\mathbb{k}$-linear, and fusion. Define its set of simple objects to be
the set $I$ of simple $C$-objects up to isomorphism. Denote
$0 \in I$ so that the monoidal identity $\mathbb{1} \in 0$. As
taking monoidal dual preserves simplicity, for each $i \in I$
there is a unique element $i^{\star}$ in $I$ such that
$V_{i}^{\star} \in i^{\star}$ whenever $V_{i} \in i$. So $I$ is a
finite set with an involution $I \xrightarrow{\star} I$. Using
the spherical structure, we can define for each $i \in I$ the
number $dim_{C}(i) = dim(i) \in \mathbb{k}$ as the trace of
$id_{V_{i}}$ and the number $\nu_{i} \in \mathbb{k}^{\star}$ as
the twisting coefficient $tr(\theta_{V_{i}})/tr(id_{V_{i}})$,
where $\theta_{V_{i}}$ denotes the endomorphism of $V_{i}$
depicted in the following graph.
\begin{center}
  \includegraphics[height=0.8cm]{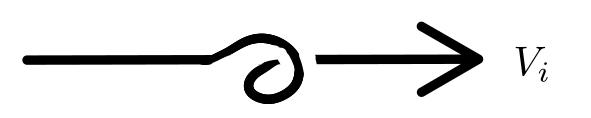}
\end{center}
We further define the Gauss sum of $C$ to be
$$\Delta_{C} = \sum_{i \in I} \nu_{i}^{-1}dim(i)^{2}.$$

\noindent In order to do computations with a premodular category
we need to choose and fix some extra data (called a coordinate).
All intrinsic results are independent of the choice (except the
square root $D$ of the global dimension).

\begin{definition}[coordinated premodular
  category]\label{def/coordinated-premodular-category}
  Let $C$ be a premodular category and $I$ its set of simple objects.
  Choose and fix the following:
  \begin{itemize}
    \item A number $D \in \mathbb{k}$ such that
          $D^{2} = \sum_{i \in I} dim_{C}(i)^{2}$ (the global
          dimension of $C$).
    \item A set of $C$-objects $\{V_{i}\}_{i \in I}$ such
          that $V_{i} \in i$ and that $V_{0} = \mathbb{1}$.
    \item A set of isomorphisms
          \cite[p.313]{turaev-qiok-3-manifolds}
          $\{\omega_{i}: V_{i} \to (V_{i^{\star}})^{\star}\}_{i \in I}$.
    \item A set of numbers
          $\{dim_{C}'(i) = dim'(i) \in \mathbb{k}\}_{i \in I}$
          such that $dim_{C}'(0) = 1$,
          $dim_{C}'(i)^{2} = dim_{C}(i)$, and
          $dim_{C}'(i^{\star}) = dim_{C}'(i)$.
    \item A set of numbers
          $\{\nu_{i}' \in \mathbb{k}\}_{i \in I}$ such that
          $\nu_{0}' = 1$, $(\nu_{i}')^{2} = \nu_{i}$, and
          $\nu_{i^{\star}}' = \nu_{i}'$
          \cite[p.313]{turaev-qiok-3-manifolds}.
  \end{itemize}

  Such a $5$-tuple
  $\vec{d} = (D, \{V_{i}\}, \{\omega_{i}\}, \{dim'(i)\}, \{\nu'_{i}\})$
  is called a coordinate of the premodular category $C$. Such a
  pair $(C, \vec{d})$ is called a coordinated premodular
  category.
\end{definition}

\noindent We will often confuse a premodular category with a
coordinated premodular category.

\begin{definition}[multiplicity module]\label{def/multiplicity-module}
  Let $C$ be a coordinated premodular category and $I$ its set
  of simple objects. Respectively, define $H^{ijk}$, $H_{k}^{ij}$, and
  $H_{ij}^{k}$ to be the $\mathbb{k}$-modules
  $Hom_{C}(\mathbb{1}, V_{i} \otimes V_{j} \otimes V_{k})$,
  $Hom_{C}(V_{k}, V_{i} \otimes V_{j})$, and
  $Hom_{C}(V_{i} \otimes V_{j}, V_{k})$.
\end{definition}

\noindent We identify $H_{k}^{ij}$ with $H^{ijk^{\star}}$ and
$H_{ij}^{k}$ with $H^{kj^{\star}i^{\star}}$ by the linear maps
induced by the following graph and call them the canonical
identifications:
\begin{center}
  \includegraphics[height=5.5cm]{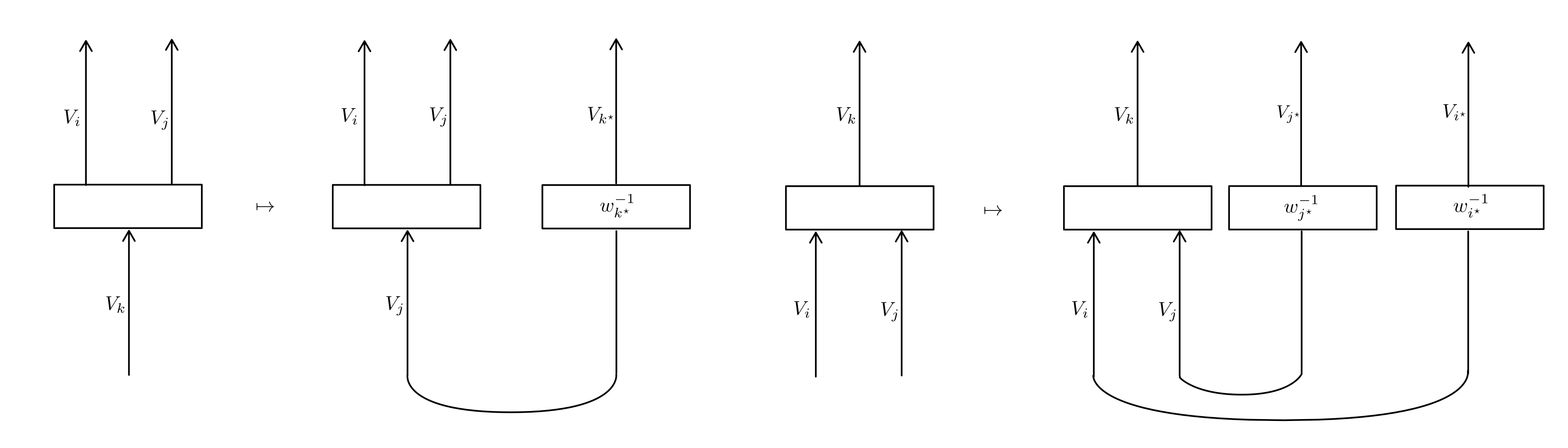}
\end{center}
\noindent Recall that the natural pairing
$$H^{ij}_{k} \otimes_{\mathbb{k}} H_{ij}^{k} \to Hom_{C}(V_{k},V_{k}) \xrightarrow{tr} \mathbb{k})$$
is nondegenerate by the semisimplicity of $C$. The braided
structure of $C$ guarantees that the $\mathbb{k}$-modules
$H^{ijk}$, $H^{ikj}$, $H^{jik}$, $H^{jki}$, $H^{kij}$, $H^{kji}$
are all isomorphic. In category theory, we must carefully
distinguish equalities from isomorphicities, hence we introduce a
way to keep track of the isomorphisms among the $H^{ijk}$'s.

\begin{definition}[canonical isomorphisms]\label{def/canonical-isomorphism}
  Let $C$ be a premodular category, $c$ its braided structure,
  $I$ its set of simple objects, and $i,j,k \in I$. Define the
  canonical isomorphisms
  $H^{ijk} \xrightarrow{\sigma_{1}(ijk)} H^{jik}$ and
  $H^{ijk} \xrightarrow{\sigma_{2}(ijk)} H^{ikj}$ by
  $$\sigma_{1}(ijk): \phi \mapsto \nu_{i}'\nu_{j}'(\nu_{k}')^{-1}(c_{V_{i}, V_{j}} \otimes id_{V_{k}})\phi,$$
  $$\sigma_{2}(ijk): \phi \mapsto \nu_{j}'\nu_{k}'(\nu_{i}')^{-1}(id_{V_{i}} \otimes c_{V_{j}, V_{k}})\phi.$$
\end{definition}

\noindent It is a simple exercise in the theory of tensor
categories to check that
\begin{equation} \label{eq1}
  \begin{split}
    \sigma_{1}(jik)\sigma_{1}(ijk) & = id, \\
    \sigma_{2}(ikj)\sigma_{2}(ijk) & = id, \\
    \sigma_{1}(jki)\sigma_{2}(jik)\sigma_{1}(ijk) & = \sigma_{2}(kij)\sigma_{1}(ikj)\sigma_{2}(ijk)
  \end{split}
\end{equation}
so $\sigma_{1}$ and $\sigma_{2}$ specify the isomorphisms among
the six $\mathbb{k}$-modules.

\begin{definition}[symmetrized multiplicity module]\label{def/symmetrized-multiplicity-module}
  Let $C$ be a premodular category, $I$ its set of simple objects, and
  $i, j, k \in I$. Define the symmetrized multiplicity module
  $H(i,j,k)$ to be the $\mathbb{k}$-module consisting of
  functions $\phi$ that assign an element
  $\phi^{i_{1}i_{2}i_{3}} \in H^{i_{1}i_{2}i_{3}}$ to each
  ordering $(i_{1}, i_{2}, i_{3})$ of the set $\{i, j, k\}$.
\end{definition}

\noindent The point is that all the symmetrized modules
$H(i,j,k)$, $H(i,k,j)$, $H(j,i,k)$, $H(j,k,i)$, $H(k,i,j)$,
$H(k,j,i)$ are \textit{equal as sets}. By definition, there is a
canonical identification between $H(i,j,k)$ and $H^{ijk}$.

\begin{definition}[contraction]\label{def/contraction}
  Let $C$ be a coordinated premodular category, $I$ its set of
  simple objects, and $i, j, k \in I$. Define the contraction map
  $H^{ijk} \otimes H^{k^{\star}j^{\star}i^{\star}} \to \mathbb{k}$
  by the following diagram \details{\cite[figure
    VI.3.5]{turaev-qiok-3-manifolds}}
  \begin{center}
    \includegraphics[height=6cm]{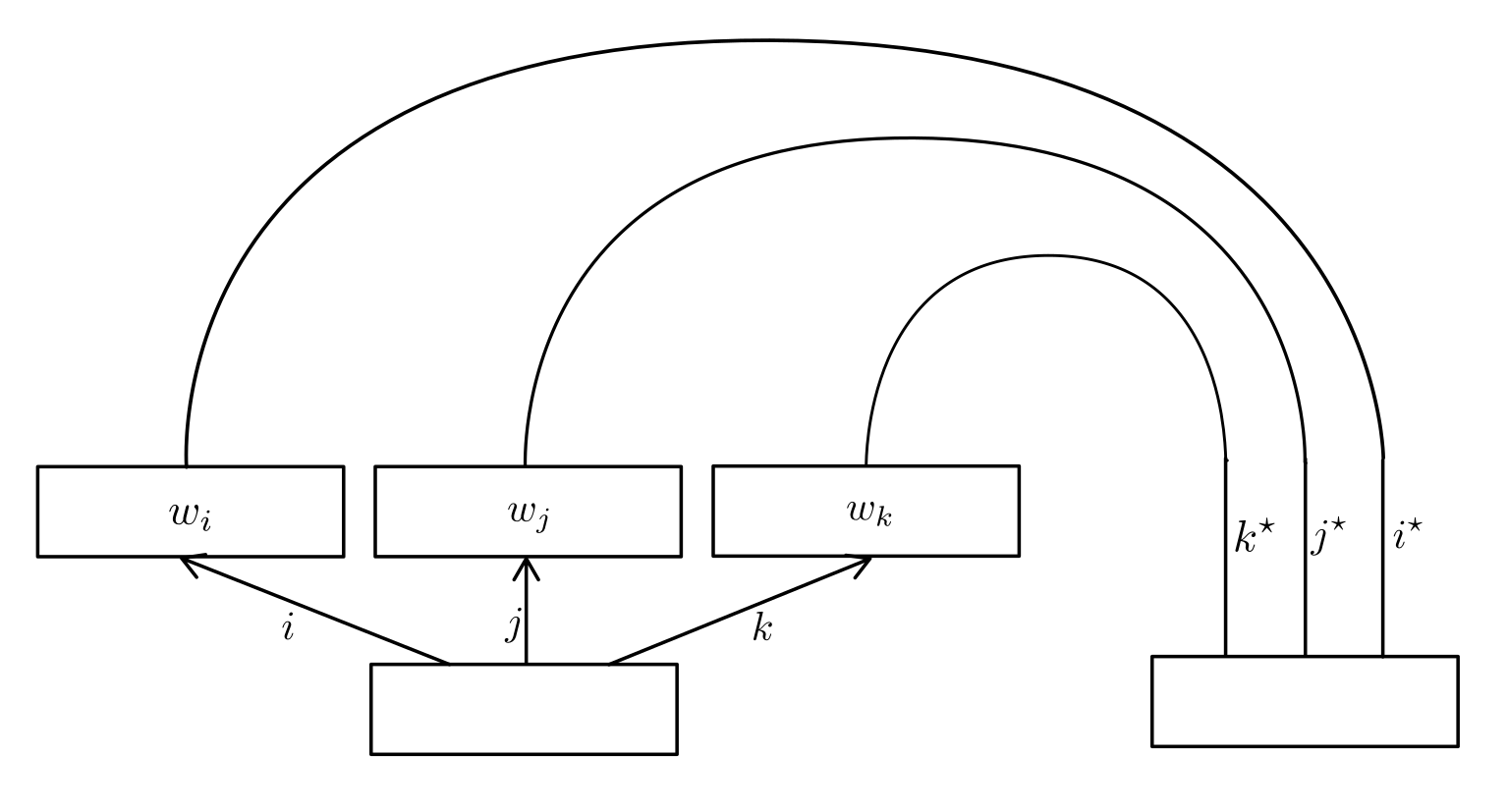}
  \end{center}
  Denote the canonically induced contraction map on the
  symmetrized modules to be
  (\cite[p.334]{turaev-qiok-3-manifolds})
  $$\ast_{ijk}: H(i,j,k) \otimes_{\mathbb{k}} H(i^{\star}, j^{\star}, k^{\star}) \to \mathbb{k}.$$
  This defines a nondegenerate pairing and thus induces a
  canonical element $Id(i,j,k)$ in the domain of $\ast_{ijk}$
  (\cite[p.333]{turaev-qiok-3-manifolds}).
\end{definition}

\noindent We will abuse notation by denoting natural contractions
from the non-ordered tensor products
$V \otimes_{\mathbb{k}} H(i,j,k) \otimes_{\mathbb{k}} H(i^{\star}, j^{\star}, k^{\star})$
to $\mathbb{k}$ by $\ast_{ijk}$ for any $\mathbb{k}$-module $V$.

\subsection{$6$j-symbol, $10$j-symbol, and $15$j-symbol}

\newcommand{\sixJSymbol}[6]{\begin{bmatrix}
  #1 & #2 & #3 \\
  #4 & #5 & #6 \\
\end{bmatrix}}

\newcommand{\normalizedSixJSymbol}[6]{\begin{vmatrix}
  #1 & #2 & #3 \\
  #4 & #5 & #6 \\
\end{vmatrix}}

% Used as in \tenJSymbol{a}{b}{c}{d}{e}{f}{g}{{h}{i}{j}}
\newcommand{\tenJSymbol}[8]{\begin{vmatrix}
    #1 & #2 & #3 & #4 \\
     . & #5 & #6 & #7 \\
    \tenJSymbolContinued #8 \\
  \end{vmatrix}_{10j}
}

\newcommand\tenJSymbolContinued[3]{
  . & . & #1 & #2 \\
  . & . & . & #3
}

\newcommand{\tenJSymbolMirrored}[8]{\begin{vmatrix}
    #1 & #2 & #3 & #4 \\
     . & #5 & #6 & #7 \\
    \tenJSymbolMirroredContinued #8 \\
  \end{vmatrix}_{\overline{10j}}
}

\newcommand\tenJSymbolMirroredContinued[3]{
  . & . & #1 & #2 \\
  . & . & . & #3
}

\begin{definition}[$6$j-symbol]\label{def/6j-symbol}
  For each $i,j,k,l,m,n \in I$, we define the $6$j-symbol
  $$\sixJSymbol{i}{j}{k}{l}{m}{n}:
  H_{k}^{ij} \otimes H_{m}^{kl} \otimes H_{jl}^{n} \otimes H_{in}^{m} \to \mathbb{k}$$
  to be the linear map induced by the partial tensor network on
  the $2$-sphere $S^{2}$:
  \begin{center}
    \includegraphics[height=6cm]{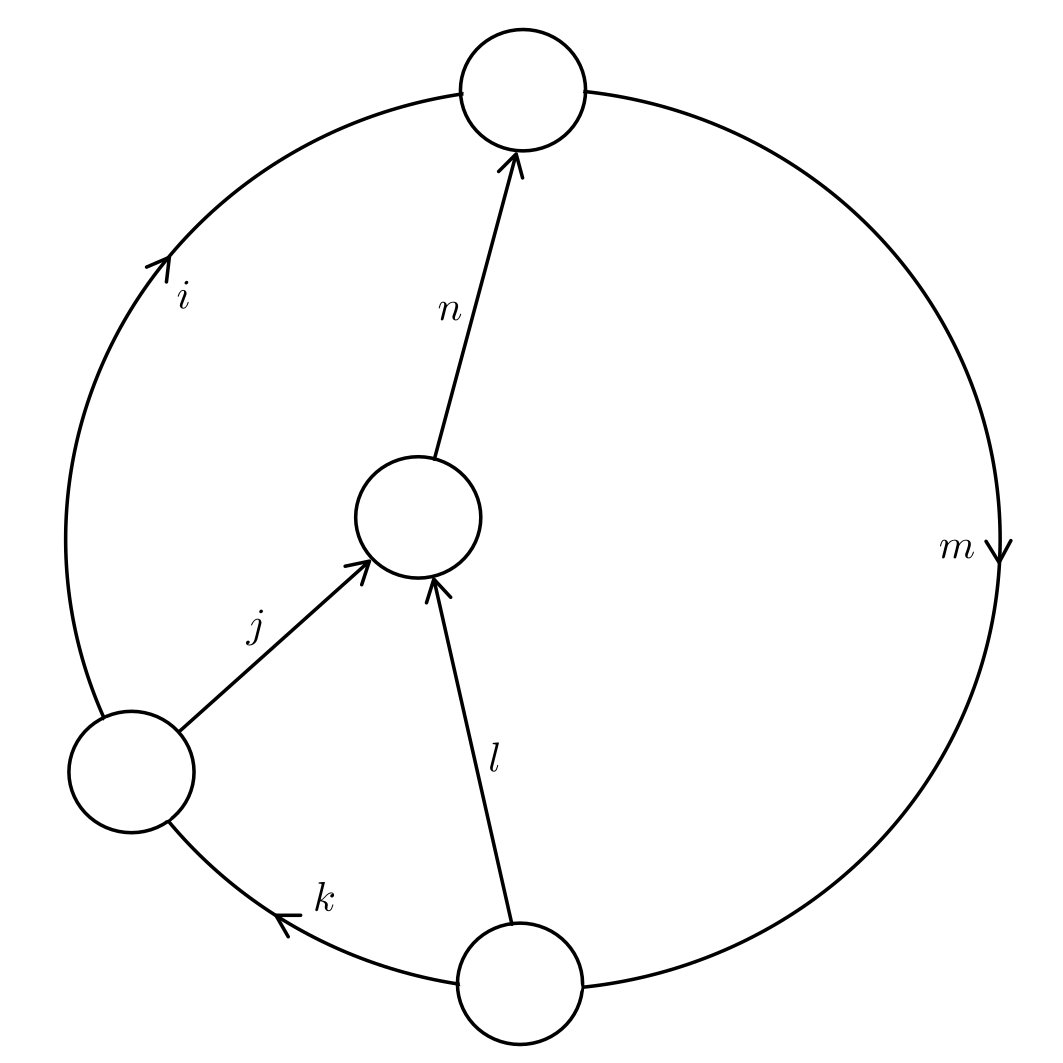}
  \end{center}

  \noindent Using the canonical identifications, we define the
  induced map
  $$\normalizedSixJSymbol{i}{j}{k}{l}{m}{n}: H(i,j,k^{\star}) \otimes H(k,l,m^{\star}) \otimes H(n,l^{\star},j^{\star}) \otimes H(m,n^{\star},i^{\star}) \to \mathbb{k}$$
  to be the normalized $6$j-symbol.
\end{definition}

\begin{proposition}[basic equalities of $6$j symbols]\label{prop/basic-equalities-of-6j-symbols}
  Let $C$ be a coordinated premodular category, $I$ its set of
  simple objects, $i, j, k, k', l, m \in I$,
  $j_{0}, j_{1}, \ldots, j_{8} \in I$, and $\delta$ be the Kronecker
  delta. Then we have the degenerated $6$j symbol
  \begin{equation}\label{eqn/degenerated-6j}
    \normalizedSixJSymbol{i}{j}{k}{l}{m}{0} = \delta_{m,i} \delta_{l,j^{\star}}dim'(i)^{-1}dim'(j)^{-1}Id(i,j,k^{\star})
    \in H(i,j,k^{\star}) \otimes_{\mathbb{k}} H(i^{\star}, j^{\star}, k)
    .
  \end{equation}
  We also have the so called Biedenharn-Elliott identity as an
  equality in the non-ordered tensor product of the
  $\mathbb{k}$-modules
  $$
  H(j_{3}^{\star}, j_{5}^{\star}, j_{6}) \otimes
  H(j_{1}^{\star}, j_{2}^{\star}, j_{5}) \otimes
  H(j_{4}^{\star}, j_{6}^{\star}, j_{0}) \otimes
  H(j_{0}^{\star}, j_{1}, j_{7}) \otimes
  H(j_{7}^{\star}, j_{2}, j_{8}) \otimes
  H(j_{8}^{\star}, j_{3}, j_{4})
  $$
  (in the context of state sum over a triangulation, this
  corresponds to the Pachner $(2,3)$-move):
  \begin{equation}\label{eqn/Biedenharn-Elliot-identity}
    \ast_{j_{0}^{\star}j_{5}j_{8}}
    \left(
      \normalizedSixJSymbol{j_{5}}{j_{3}}{j_{6}}{j_{4}}{j_{0}}{j_{8}} \otimes \normalizedSixJSymbol{j_{1}}{j_{2}}{j_{5}}{j_{8}}{j_{0}}{j_{7}}
    \right)
    =
    \sum_{j \in I} dim(j)
    \ast_{j^{\star}j_{2}j_{3}}\ast_{jj_{4}j_{7}^{\star}}\ast_{jj_{1}j_{6}^{\star}}
    \left(
      \normalizedSixJSymbol{j_{1}}{j_{2}}{j_{5}}{j_{3}}{j_{6}}{j} \otimes
      \normalizedSixJSymbol{j_{1}}{j}{j_{6}}{j_{4}}{j_{0}}{j_{7}} \otimes
      \normalizedSixJSymbol{j_{2}}{j_{3}}{j}{j_{4}}{j_{7}}{j_{8}}
    \right).
  \end{equation}
  We also have the orthonormality relation
  \begin{equation}\label{eqn/orthonormality-relation}
    \delta_{k,k'} Id(i,j,k^{\star}) \otimes Id(k,l,m^{\star})
    =
    dim(k)\,\sum_{n \in I} dim(n) \ast_{im^{\star}n} \ast_{jln^{\star}}
    \left(
      \normalizedSixJSymbol{i^{\star}}{j^{\star}}{k^{\star}}{l^{\star}}{m^{\star}}{n^{\star}} \otimes
      \normalizedSixJSymbol{i}{j}{k'}{l}{m}{n}
    \right).
  \end{equation}
  Finally, we have the Racah identity
  \begin{equation}\label{eqn/Racah-identity}
    \nu'_{j_{3}} \nu'_{j_{6}} (\nu'_{j_{1}} \nu'_{j_{2}} \nu'_{j_{4}} \nu'_{j_{5}})^{-1}
    \normalizedSixJSymbol{j_{1}}{j_{2}}{j_{3}}{j_{4}}{j_{5}}{j_{6}}
    =
    \sum_{j \in I} (\nu'_{j})^{-1} dim(j) \ast_{j^{\star}j_{1}j_{4}} \ast_{jj_{2}j_{5}^{\star}}
    \left(
      \normalizedSixJSymbol{j_{1}}{j_{4}}{j}{j_{2}}{j_{5}}{j_{6}} \otimes
      \normalizedSixJSymbol{j_{2}}{j_{1}}{j_{3}}{j_{4}}{j_{5}}{j}
    \right)
  \end{equation}
\end{proposition}

\begin{proof}
  Proofs and references for the modular case can be found in
  \cite[section VI.5.4]{turaev-qiok-3-manifolds}. The proof does
  not use modularity at all, so it carries through for the
  premodular case verbatim.
\end{proof}

\begin{definition}[$10$j-symbol]\label{def/10j-symbol}
  Let $C$ be a coordinated premodular category, $I$ its set of
  simple objects, and $j_{ab} \in I$ with $j_{ab} = j_{ba}^{\star}$ for
  $0 \leq a, b \leq 4$.
  Denote $[x,y,z,w]$ to be the vector space
  $Hom_{C}(V_{0}, V_{j_{x}} \otimes V_{j_{y}} \otimes V_{j_{z}} \otimes V_{j_{w}})$.
  Then define the $10$j symbol (and its mirror, resp.)
  $$
  \tenJSymbol{j_{01}}{j_{02}}{j_{03}}{j_{04}}{j_{12}}{j_{13}}{j_{14}}{{j_{23}}{j_{24}}{j_{34}}} \quad \left(\tenJSymbolMirrored{j_{01}}{j_{02}}{j_{03}}{j_{04}}{j_{12}}{j_{13}}{j_{14}}{{j_{23}}{j_{24}}{j_{34}}}, resp.\right)
  $$
  to be the $\mathbb{k}$-linear map from the non-ordered tensor
  product of $\mathbb{k}$-modules
  $$
  [01,02,03,04] \otimes
  [12,13,14,10] \otimes
  [23,24,20,21] \otimes
  [34,30,31,32] \otimes
  [40,41,42,43]
  $$
  % \begin{itemize}
  %   \item
  %         $Hom_{C}(V_{0}, V_{01} \otimes V_{02} \otimes V_{03} \otimes V_{04})$
  %   \item
  %         $Hom_{C}(V_{0}, V_{12} \otimes V_{13} \otimes V_{14} \otimes V_{01}^{\star})$
  %   \item
  %         $Hom_{C}(V_{0}, V_{23} \otimes V_{24} \otimes V_{02}^{\star} \otimes V_{12}^{\star})$
  %   \item
  %         $Hom_{C}(V_{0}, V_{34} \otimes V_{03}^{\star} \otimes V_{13}^{\star} \otimes V_{23}^{\star})$
  %   \item
  %         $Hom_{C}(V_{0}, V_{04}^{\star} \otimes V_{14}^{\star} \otimes V_{24}^{\star} \otimes V_{34}^{\star})$
  % \end{itemize}
  to $\mathbb{k}$ induced by the following (equivalent)
  $C$-colored graphs (resp., the same gadget but with the
  underlying graph mirrored and all arrows reversed).
  \begin{center}\label{graph/10j-symbol}
    \includegraphics[height=4cm]{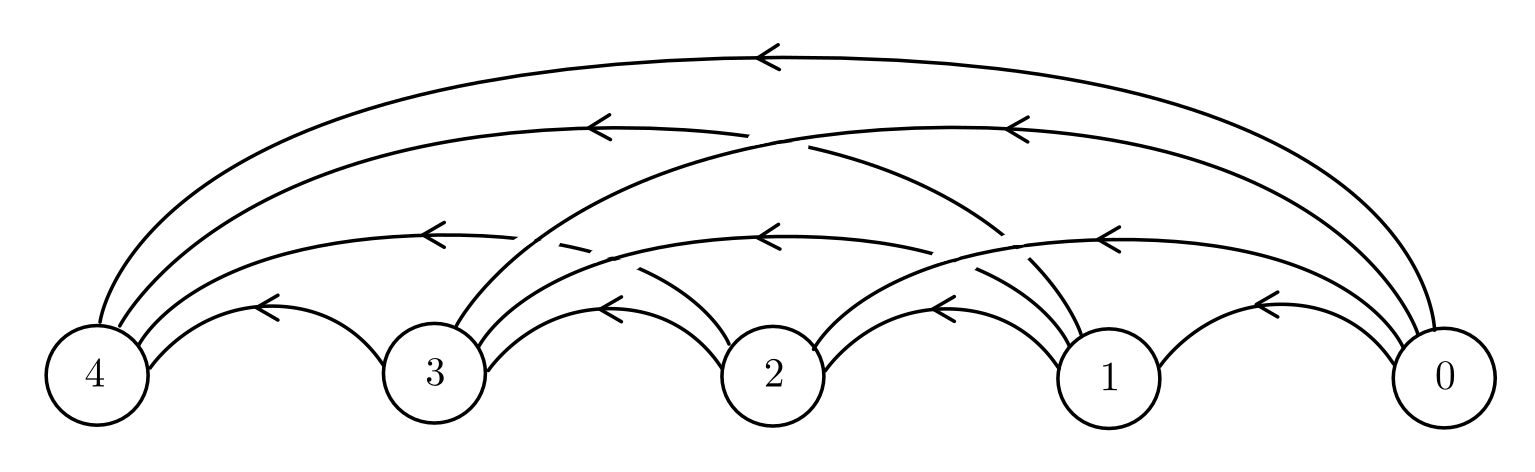}
    \includegraphics[height=6cm]{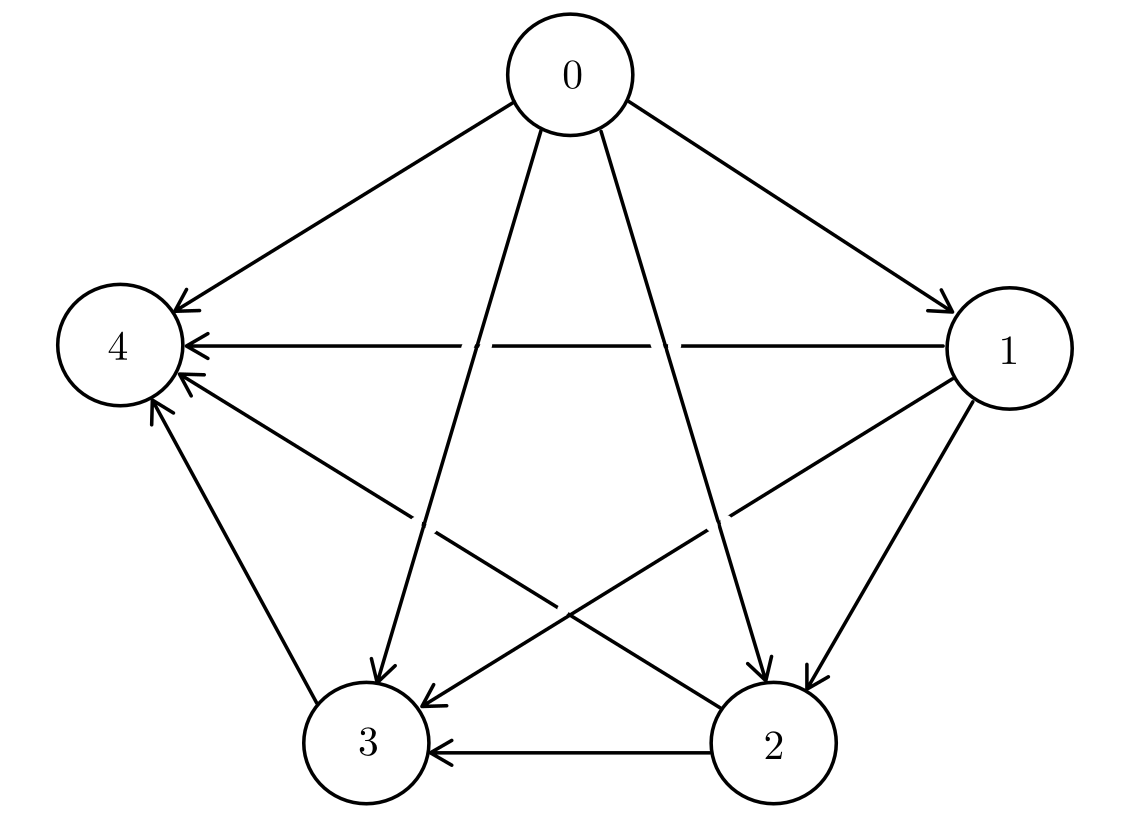}
  \end{center}
\end{definition}

\begin{remark}[$15$j-symbol]\label{remark/15j-symbol}
  A $15$j-symbol is an equivalent variant of a $10$j-symbol. It
  was used in the older literature to make sure the morphism
  spaces are $1$-dimensional. The $10$j-symbols are more
  intrinsic, so we use them instead of the $15$j-symbols.
\end{remark}

\section{Topology $(T)$}\label{section/topology}
\subsection{$4$-manifold}

Manifolds in real dimension $4$ are interesting because of their
wildness, witnessed in the following examples:

\begin{enumerate}
  \item Real dimension $4$ is the smallest dimension where the
        topological structures and the smooth structures
        disagree.
  \item For $n \in \mathbb{N}\setminus\{4\}$, the euclidean space
        $\mathbb{R}^{n}$ as a topological space admits exactly
        one diffeomorphism type, while $\mathbb{R}^{4}$ admits
        infinitely many
        \cite{scorpan/the-wild-world-of-4-manifolds}\cite[p.2]{milnor/topological-manifolds-and-smooth-manifolds}.
  \item The (smooth) Poincare conjecture for the $n$-dimensional
        sphere $S^{n}$ has been resolved except for $n=4$, which
        remains widely open to date despite several attempts.
  \item The Universe where we live seems to be best-modeled by a
        $4$-manifold.
\end{enumerate}

\noindent Despite its wildness, in dimension $4$ the notion of
smooth manifolds coincides with the notion of piecewise-linear
(PL) manifolds \cite[sec.IX.1.1]{turaev-qiok-3-manifolds}. The
data of the later can be made combinatorial and concrete, and is
what we will really be working on. From now on, unless further
specified, by a manifold we mean an oriented, connected, closed
and piecewise-linear manifold in real dimension $4$.

\subsection{Triangulation} \label{subsection/triangulation}

This section is standard
\cite[chap.1]{rourke-sanderson/intro-to-pl-topology}
\cite[sec.2]{manolescu/lectures-on-the-triangulation-conjecture}
but included for completeness.

\begin{definition}[simplicial complex]\label{def/simplicial-complex}
  An abstract simplicial complex is a pair $K = (V, S)$ of finite
  sets $V$ and $S \subset 2^{V}$, such that $\tau \in S$ whenever
  $\sigma \in S$ and $\tau \subset \sigma$. For a subset
  $S' \subset S$, its closure is
  $$\overline{S'} = \{\tau \in S \,|\, \tau \subset \sigma \in S'\}.$$
  Given a simplex $\tau$, its star and its link are
  $$Star(\tau) := \{\sigma \in S \,|\, \tau \subset \sigma\}, \quad
  Link(\tau) := \{\sigma \in \overline{Star(\tau)} \,|\, \tau \cap \sigma = \phi\}.$$
  We say that $K$ is an abstract combinatorial manifold (possibly
  with boundary) of dimension $n$ if the link of each of its
  simplices (or equivalently each of its vertices) is PL
  homeomorphic to either a sphere or a disk, and if top cell has
  dimension $n$. The geometric realization $|K|$ of $K$ is
  defined as usual by gluing $k$-dimensional simplices inducively
  on $k \geq 0$.

  An orientation of a combinatorial manifold is an ordering of
  the vertices up to even permutations. We define the standard
  $n$-simplex to be $\Delta_{n} = \{0, 1, 2, \ldots, n\}$, with
  $[0 < 1 < \ldots < n]$ being its standard orientation. Its
  $k$th face is defined to be
  $$\Delta_{n}(\widehat{k}) = (\text{-} 1)^{k}\Delta(012 \ldots \widehat{k} \ldots n).$$
  For example, the standard oriented $4$-simplex
  $\Delta_{4}(01234)$ has a $3$-dimensional face being
  $\Delta_{4}(\widehat{1}) = \text{-} \Delta_{4}(0234) = \Delta_{4}(2034) = \ldots$.
  This face, in turn, has another $2$-dimensional face
  $\Delta_{4}(\widehat{12}) = \Delta_{4}(034)$. In general, for
  $i<j$, denote
  $\Delta(\widehat{ij}) = (\text{-} 1)^{i+j \text{-} 1}(0 \ldots \widehat{i} \ldots \widehat{j} \ldots 4)$.
\end{definition}

\begin{definition}[Pachner move]\label{def/pachner-move}
  Let an abstract combinatorial manifold $K = (V,S)$ of
  dimension $n$. A Pachner $(1,n+1)$-move along a top-simplex
  $\tau \in S$ is defined to be $K \rightsquigarrow K'$, where
  $$
  K' = \left(
    V \coprod \{\star\},\quad (S\setminus\{\tau\}) \coprod (\coprod_{f} (f \cup \{\star\}))
  \right),
  $$
  where $f$ funs through each face of $\tau$. A Pachner
  $(2,n)$-move along two top-simplices $\tau, \tau' \in S$ that
  share a face $f \in S$ is defined to be
  $K \rightsquigarrow K'$, where
  $$ K' = \left (V,\quad (S\setminus\{\tau,\tau'\}) \coprod (\coprod_{g} (g \cup \{\star, \star'\})) \right),$$
  where $g$ runs through each face of $f$, and $\star$ ($\star'$,
  resp.) denotes the opposite vertex of $f$ in $\tau$ ($\tau'$,
  resp.). We say the inverses are Pachner $(n,2)$-moves and
  $(n+1,1)$-moves respectively. Denote $K \sim K'$ if $K'$ can be
  obtained by $K$ via a finite sequence of Pachner moves.
  \begin{center}
    \includegraphics[height=4cm]{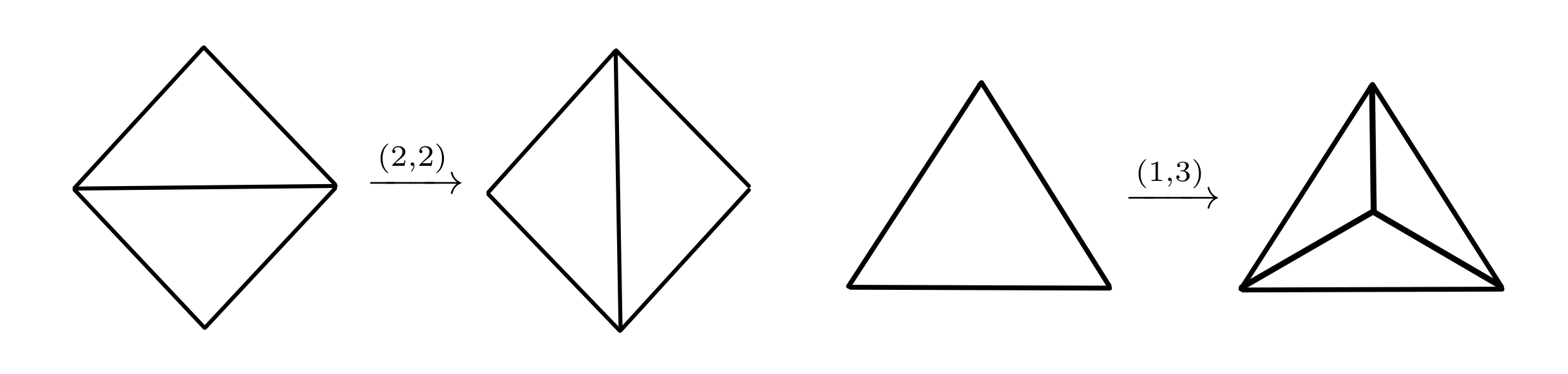}
  \end{center}
\end{definition}

\noindent Notice that a Pachner move $K \rightsquigarrow K'$
induces naturally a PL homeomorphism $K \xrightarrow{\sim} K'$.

\begin{definition}[triangulation of PL-manifolds] Let $X$ be a
  piecewise-linear manifold. A triangulation of $X$ is a
  PL-homeomorphism $X \xrightarrow{\phi} |K|$ for some
  combinatorial manifold $K$.
\end{definition}

\begin{fact}
  Any piecewise-linear manifold $X$ has a triangulation
  $\phi: X \simeq |K|$. Any other triangulation
  $\phi': X \simeq |K'|$ satisfies $K \sim K'$. Finally, an
  orientation of $X$ restricts to a coherent orientation for each
  top cell of $K$.
\end{fact}

\subsection{Handle decomposition}

\noindent By Morse's theory of extremal points, any smooth
manifold admits a handle decomposition. By Cerf theory, two
handle decompositions present the same manifold (up to
diffeomorphism) if and only if both decomposition data are
related by a finite sequence of handle creations, handle
annihilations, and handle slides
\cite{gompf-stipsicz/4-manifolds-and-kirby-calculus}. A
triangulation of a manifold admits a natural handle decomposition
by taking dual. The correct state sum based on this datum is the
universal state sum \cite{walker/universal-state-sum}; it
transforms a handle decomposition into a number. A useful fact to
notice is that closed $4$-manifolds are reconstructible from
their handles of indices $0$, $1$, and $2$
(\ref{remark/reconstruction-of-4-manifolds}).

\subsection{Shadow}

% TODO (less urgent) - Figure out why we need integer-shadows.

\noindent A shadow is another type of structure that encodes
closed $4$-manifolds. Roughly speaking, a shadow is a
$2$-polyhedron with extra decorations (called gleams) that
remember the twisting data. A $2$-polyhedron is a topological and
combinatorial object that encodes $3$-dimensional manifolds
\cite{matveev/algorithmic-topology-and-classification-of-3-manifolds}.
It is called a pre-foam in the literature of Khovanov homology
(from foams) \cite{khovanov-robert/foam}.

\begin{definition}[tripod]\label{def/tripod}
  Define the standard tripod to be the topological subspace of
  $\mathbb{R}^{3}$ consisting of the points $(x,y,z)$ such that
  at least two of the entries are zero, and the last entry
  belongs to $[0,1)$. Define a tripod to be any topological space
  homeomorphic to the standard tripod.
\end{definition}

\begin{definition}[cone]\label{def/cone}
  For each topological space $X$, define its standard open cone
  $cone(X)$ to be the quotient space
  $(X \times \mathbb{R_{\geq 0}})/((x,0) \sim (x',0)).$ Define an
  open cone of $X$ to be any topological space homeomorphic to
  $cone(X)$.
\end{definition}

\begin{definition}[local shape]\label{def/local-shape}
  Let $X$ be a topological space and $x \in X$. Denote by $T$ the
  standard tripod and $S$ the $1$-skeleton of the boundary of the
  standard tetrahedron (a trivalent graph with $4$ vertices and
  $6$ edges). Respectively, we say that $x$ is a smooth point, a
  line point, a tetrahedral point, a boundary smooth point, or a
  boundary line point of $X$ if it has a relative neighborhood
  homeomorphic to $(\mathbb{R}^{2},0)$,
  $(T \times \mathbb{R}, (0, 0))$, $(cone(S), (*, 0))$,
  $(\mathbb{R} \times \mathbb{R}_{\geq 0}, (0, 0))$, or
  $(T \times \mathbb{R}_{\geq 0}, (0, 0))$.
\end{definition}

\begin{definition}[simple $2$-polyhedron]\label{def/simple-2-polyhedron}
  A simple $2$-polyhedron with boundary is defined to be a
  piecewise-linear compact CW-complex $P$ of real dimension two,
  such that each of its point $p$ is either a smooth point, a
  line point, a tetrahedral point, a boundary smooth point, or a
  boundary line point. If only the first three types are
  involved, we call $P$ a simple $2$-polyhedron without boundary.
\end{definition}

\begin{definition}[components of a simple $2$-polyhedron]\label{def/components-of-a-simple-2-polyhedron}
  Let $P$ be a simple $2$-polyhedron with boundary. Define the
  set of smooth points (or called interior points) of $P$ to be
  $Int(P)$. Define the set of line points, tetrahedral points,
  and boundary line points to be $sing(P)$. Define the set of
  boundary line points and boundary smooth points to be
  $\partial P$. Call a connected component of $Int(P)$ to be a
  region of $P$; define the set of regions to be $Region(P)$. $P$
  is said to be orientable if each region of $P$ is orientable.
  An orientation of $P$ is an assignment of orientations to each
  of the region.
\end{definition}

\begin{figure}[ht]
\centering
\begin{tikzpicture}
\input{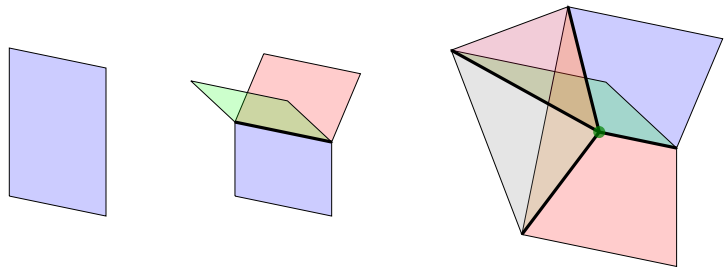}
\end{tikzpicture}
\caption{The graphic is taken from \cite{khovanov-robert/foam}.}
\label{figure/3localmodel}
\end{figure}

\begin{definition}[shadowed $2$-polyhedron]\label{def/shadowed-2-polyhedron}
  Let $P$ be a simple $2$-polyhedron, and $A$ an abelian group
  with a distinguished element $\omega \in A$. We define a shadow
  to be a pair of an orientable $2$-polyhedron $P$ and a map
  (called gleam) $gl: Region(P) \to A$. Unless specified further,
  we assume that $A = \mathbb{Z}\left[\frac{1}{2}\right]$ and
  $\omega = \frac{1}{2}$. We denote $-P$ to be the same simple
  $2$-polyhedron but with all gleams flipped by $(a \mapsto -a)$.
\end{definition}

\noindent For each connected oriented closed surface $\Sigma$ and
each $a \in A$, there is a shadowed $2$-polyhedron $\Sigma_{a}$
which consists of $\Sigma$ with the gleam $a$ assigned to the
only region. For example, $S^{2}_{0}$ denotes the $0$-gleamed
$2$-sphere.

\begin{definition}[nullity of a shadowed $2$-polyhedron]\label{def/nullity-of-a-shadowed-2-polyhedron}\cite[section VIII.5.1]{turaev-qiok-3-manifolds}
  Let $P$ be an oriented shadowed $2$-polyhedron. For each region
  $Y$ of $P$, the contraction map
  $(P / \partial P) \to P / (P \setminus Y)$ and the orientation
  of $Y$ induces a map
  $$H_{2}(P;\partial P) \to \mathbb{Z}; h \mapsto \langle h | Y \rangle.$$
  Define the symmetric bilinear form $\tilde{Q}_{P}$ on
  $H_{2}(P; \partial P)$ by summing over all regions of $Y$
  $$\tilde{Q}_{P}(h_{1}, h_{2}) = \sum_{Y} \langle h_{1}|Y\rangle \langle h_{2}|Y\rangle gl(Y) \in A$$
  and restrict it to $Q_{P}$ along the natural map
  $H_{2}(P) \to H_{2}(P; \partial P)$ (which is injective by a
  usual argument using long exact sequence). $H_{2}(P)$ is a free
  abelian group, and so is $Ann(Q_{P})$. Finally, define the
  nullity of $P$ to be $null(P) = rank(Ann(Q_{P}))$.
\end{definition}

\noindent We remark that if the shadowed polyhedron comes from a
$4$-manifold $X$, then the bilinear form defined in the previous
definition coincide with the intersection form of $X$
\cite[section IX.5]{turaev-qiok-3-manifolds}.

\begin{definition}[shadow moves]\label{def/shadow-moves}
  \details{\cite[section VIII.1.3,
    p.369]{turaev-qiok-3-manifolds}} The basic shadow moves
  $P_{1}, P_{2}, P_{3}$ are given in the following graphics
  (taken from \cite{turaev-qiok-3-manifolds}). A shadow move is a
  finite composition of the $P_{i}^{\pm 1}$'s.
  \begin{center}
    \includegraphics[height=13cm]{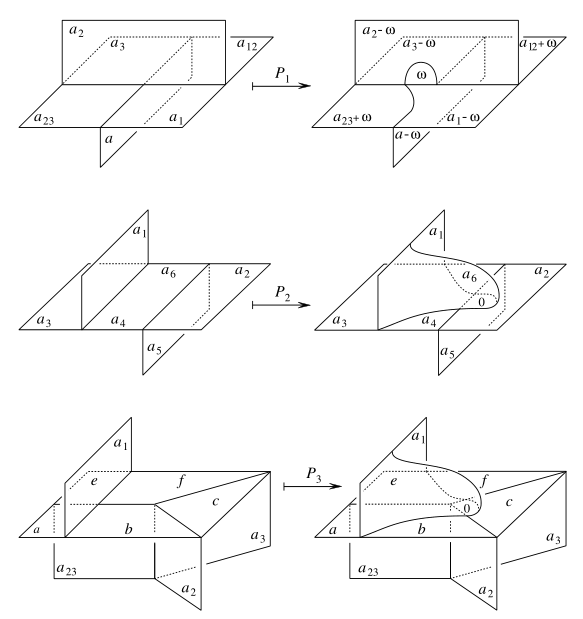}
  \end{center}
\end{definition}

\begin{definition}[shadow]\label{def/shadow}
  A shadow is an equivalence class of shadowed $2$-polyhedron $P$
  up to a shadow move. We denote the shadow by $[P]$, and say
  that $P$ represents the shadow $[P]$
  \cite[p.370]{turaev-qiok-3-manifolds}.
\end{definition}

\noindent For two connected shadow $[P]$ and $[P']$, we construct
the shadow $[P]+[P']$ as follows. Arbitrarily identify two
arbitrarily chosen closed disks $D \subset Int(P)$ and
$D' \subset Int(P')$ in $P \coprod P'$, and equip the interior of
$D$ (a new region) with gleam $0$. So defines a simple
$2$-polyhedron and we say that it represents $[P]+[P']$. It is
well-defined by \cite[lemma VIII.2.1.1]{turaev-qiok-3-manifolds}.
For an integer $m \in \mathbb{Z_{\geq 0}}$, we define $m[P]$ as
the sum of $m$-many $[P]$.

\begin{definition}[stable shadow]\label{def/stable-shadow}
  Two connected shadowed polyhedra $P$, $P'$ are called stably
  shadow equivalent if there exists
  $n, n' \in \mathbb{Z_{\geq 0}}$ such that
  $[P] + m[S^{2}_{0}] = [P'] + m'[S^{2}_{0}]$. Extend the
  definition to non-connected ones in an obvious fashion. A
  stable shadow is defined to be a shadowed polyhedron up to
  stable shadow equivalence. Denote the stable shadow of $[P]$ to
  be $stab([P])$.
\end{definition}

\noindent We are ready to present a closed $4$-manifold in terms
of shadows.

\begin{definition}[locally flat $2$-polyhedron in a
  $4$-manifold]\label{def/locally-flat-2-polyhedron-in-a-4-manifold}
  Let $X$ be a closed $4$-manifold. A $2$-polyhedron $P$ in $X$
  is flat at a point $p \in P$ if there exists a neighborhood $U$
  of $p$ in $X$ such that $U \cap P$ lies in a $3$-dimensional
  submanifold of $X$. We say that $P$ is locally flat if it is
  flat at all $p \in P$ (\cite[p.394]{turaev-qiok-3-manifolds}).
\end{definition}

\begin{definition}[skeleton of a $4$-manifold]\label{def/skeleton-of-a-4-manifold}
  Let $X$ be a closed $4$-manifold. A skeleton
  \cite[p.395]{turaev-qiok-3-manifolds} of $X$ is a locally flat
  orientable simple $2$-polyhedron without boundary $P$ such that
  a closed regular neighborhood of it with some $3$- and
  $4$-handles form $X$.
\end{definition}

\noindent For example, $\mathbb{C}P^{1} = \{[x:y:0]\}$ is a
skeleton of $\mathbb{C}P^{2} = \{[x:y:z]\}$. By \cite[theorem
IX.1.5]{turaev-qiok-3-manifolds}, every $4$-manifold has a
skeleton (by compressing the $(0,1,2)$-handles in an arbitrary
handle decomposition).

\begin{definition}[stable shadow of a $4$-manifold]\label{def/stable-shadow-of-a-4-manifold}
  Let $X$ be a closed $4$-manifold. Take a skeleton $P$ of $X$
  and construct a shadowed simple $2$-polyhedron by assigning
  gleams to the regions $\Sigma$ in the following way.
  \begin{enumerate}
    \item If $\Sigma$ is homeomorphic to a closed surface, define
          the gleam to be the self-intersection (which is
          independent of the orientation of $\Sigma$)
          $$([\Sigma] \cdot [\Sigma]) \in H_{0}(X;\mathbb{Z}) = \mathbb{Z} \subset \mathbb{Z}\left[1/2\right].$$
    \item Otherwise, $\Sigma$ is non-compact. Deformation retract
          it to a compact subsurface $\Sigma_{0}$. Denote $N$ to
          be the normal bundle of $\Sigma_{0}$ in $X$. Consider
          the line bundle $l$ over $\partial \Sigma_{0}$ by
          \cite[section VIII.6.2,
          p.397]{turaev-qiok-3-manifolds}, which may be regarded
          as a sub-bundle of $N|_{\partial \Sigma_{0}}$. The
          circle bundle $\mathbb{P}(N)$ is trivial over
          $\Sigma_{0}$ since the later is a homotopy $1$-type.
          With a choice of an orientation of $\Sigma_{0}$ and
          $X$, $l$ induces a section of
          $\mathbb{P}(N)|_{\partial}$. The obstruction class of
          this section to the whole $\mathbb{P}(N)$ is an element
          of
          $H^{2}(\Sigma_{0}, \partial \Sigma_{0}; \pi_{1}(S^{1})) = \mathbb{Z}$.
          % [TODO (less urgent): Understand this part really.]
          Finally, define the gleam to be the half of the
          resulting integer (which is independent to the choice
          of $\Sigma_{0}$).
  \end{enumerate}
  It is the main theorem of \cite[section
  IX.1.7]{turaev-qiok-3-manifolds} that all shadowed polyhedra
  chosen in such fashion above are all stably shadow equivalent.
  Therefore, it defines the stable shadow $sh(X)$ of the closed
  $4$-manifold $X$.
\end{definition}

\begin{example}
  $sh(\pm\mathbb{C}P^{2}) = stab([S^{2}_{\pm 1}])$ and
  $sh(S^{4}) = stab([S^{2}_{0}])$.
\end{example}

\noindent A handle decomposition of a closed $4$-manifold $X$
gives rise to a shadow of $X$ \cite[section
IX.4]{turaev-qiok-3-manifolds}. The explicit construction will be
recalled below in
\ref{def/shadow-of-a-4-manifold-from-a-handle-decomposition},
which will be used to prove our main theorem.

\begin{definition}[skeleton of a $3$-manifold]\label{def/skeleton-of-a-3-manifold}
  Let $Y$ be a closed $3$-manifold. A skeleton of $Y$ is an
  orientable simple $2$-polyhedron without boundary $P \subset Y$
  such that $Y \setminus P$ is a disjoint union of open
  $3$-balls \cite[p. 400]{turaev-qiok-3-manifolds}.
\end{definition}

\begin{definition}[shadow cone of a framed link in a
  $3$-manifold]\label{def/shadow-cone-of-a-framed-link-in-a-3-manifold}
  \noindent Every compact $3$-manifold $Y$ has a skeleton
  \cite[theorem IX 2.1.1]{turaev-qiok-3-manifolds}. For example,
  the equator $S^{2}$ of $S^{3}$ is a skeleton. Let $P$ be a
  skeleton of $Y$ and $l$ be a framed link in $Y$. Projecting $l$
  generically onto $P$ induces a shadow projection. Assign gleams
  around each crossing point as in \cite[figure
  IX.3.4]{turaev-qiok-3-manifolds}. Then construct the shadow by
  naturally attaching a disk along each projected component on
  $P$ (as a new region) endowed with zero gleam. Denote the
  resulting shadow to be $CO(Y,l)$ (well-defined up to stable
  shadow moves \cite[section IX.3.3]{turaev-qiok-3-manifolds}).
\end{definition}

\begin{definition}[shadow of a $4$-manifold from a handle
  decomposition]\label{def/shadow-of-a-4-manifold-from-a-handle-decomposition}
  Let $X$ be an oriented $4$-manifold and
  $H = \bigcup_{i=0}^{4} H_{i}$ be a handle decomposition, where
  $H_{i}$ denotes the union of the handles of index $i$. Define
  $Y$ to be the closed $3$-manifold $\partial(H_{0} \cup H_{1})$.
  By the definition of handle decomposition, the gluing datum of
  $H_{2}$ onto the handles with lower indices is encoded as a
  link $l$ in $Y$. Define the stable shadow $sh'(X,H)$ to be
  $CO(Y,l)$.
\end{definition}

\begin{remark}\label{remark/stable-shadow-of-a-4-manifold}
  It is a theorem of \cite[sec.IX.4.2]{turaev-qiok-3-manifolds}
  that $sh'(X,H)$ does not depend on the choice of $H$ as a
  stable shadow. In fact, $sh'(X,H)$ equals the stable shadow
  $sh(X)$ \cite[sec. IX.7]{turaev-qiok-3-manifolds}.
\end{remark}

\begin{remark}\label{remark/reconstruction-of-4-manifolds}\cite[section 4.4]{gompf-stipsicz/4-manifolds-and-kirby-calculus}
  The handles of indices $\leq 2$ are enough to reconstruct the
  whole closed $4$-manifold.
\end{remark}

\section{Sum $\left( \int_{T}{A} \right)$}\label{section/sum}
\subsection{Crane-Yetter state sum}

\noindent Throughout this section, let $C$ to be a coordinated
premodular category and $I$ be the set of simple $C$-objects.

% \begin{definition}[$C$-colored ($4$-)simplex]\label{def/C-colored-simplex}
%   Let $C$ be a coordinated premodular category, $I$ be the set
%   of simple objects in $C$, and $\Delta$ be the standard oriented
%   $4$-cell $\Delta_{4}(01234)$ defined in section
%   \ref{subsection/triangulation}. Define a $C$-coloring on
%   $\Delta$ to be a map from the set of oriented $2$-cells of
%   $\Delta$ to $I$.
%   $\beta(-\Delta_{2}) = \beta(\Delta_{2})^{\star}$. A $C$-colored
%   $4$-simplex is a $4$-simplex with a $C$-coloring.
% \end{definition}

% \begin{definition}[$C$-colored triangulation]\label{def/C-colored-triangulation}
%   Let $C$ be a coordinated premodular category, $I$ be the set
%   of simple objects in $C$, $X$ be a closed oriented $4$-manifold, and
%   $T$ be an (oriented) triangulation of $X$. A $C$-coloring of
%   $T$ is an assignment for each $4$-simplex of $T$ a
%   $C$-coloring. Define a $C$-colored triangulation to be a
%   triangulation with a $C$-coloring.
% \end{definition}

\begin{definition}[colored combinatorial manifold]
  A $C$-coloring of a combinatorial manifold $X$ is a map
  $\beta: X_{2} \to I$, where $X_{2}$ denotes the set of oriented
  $2$-simplices of $X$, such that
  $\beta(\text{-}x) = \beta(x)^{\star}$ for all $x \in X_{2}$. A
  $C$-colored combinatorial manifold is a pair of a combinatorial
  manifold and a $C$-coloring of X.
\end{definition}

\begin{definition}[$10$j symbol for a colored simplex]\label{def/10j-symbol-for-a-C-colored-simplex}
  Let $\Delta$ be a $4$-simplex with a total ordering on the set
  of vertices, and let $\beta$ to be a $C$-coloring for $\Delta$.
  $C$-colored simplex. Denote $\beta_{\widehat{ab}}$ to be the
  color $\beta(\Delta_{4}(\widehat{ab})) \in I$ assigned to the
  oriented $2$-cell $\Delta_{4}(\widehat{ab})$. We define the
  $10$j symbols for $(\Delta,\beta)$ to be the $10$j-symbols
  (\ref{def/10j-symbol})
  $$
  10j(\Delta) = \tenJSymbol{\beta_{\widehat{01}}}{\beta_{\widehat{02}}}{\beta_{\widehat{03}}}{\beta_{\widehat{04}}} {\beta_{\widehat{12}}}{\beta_{\widehat{13}}}{\beta_{\widehat{14}}} {{\beta_{\widehat{23}}}{\beta_{\widehat{24}}}{\beta_{\widehat{34}}}}, \quad
  \overline{10j}(\Delta) = \tenJSymbolMirrored{\beta_{\widehat{01}}}{\beta_{\widehat{02}}}{\beta_{\widehat{03}}}{\beta_{\widehat{04}}} {\beta_{\widehat{12}}}{\beta_{\widehat{13}}}{\beta_{\widehat{14}}} {{\beta_{\widehat{23}}}{\beta_{\widehat{24}}}{\beta_{\widehat{34}}}}.
  $$
\end{definition}

\begin{definition}[Crane-Yetter state sum for a closed
  $4$-manifold]\label{def/crane-yetter-state-sum-for-a-closed-4-manifold}
  Let $X$ be an connected, oriented, closed piecewise-linear
  manifold, $\phi: X \xrightarrow{\sim} |K|$ a triangulation,
  $\beta: K_{2} \to I$ a $C$-coloring of $K$, and $\tau$ a total
  ordering on the set of vertices of $K$.

  For each $4$-simplex $\Delta$ of $K$, we assign a $10$j-symbol
  $10J(\beta,\Delta)$ as follows. If the orientation restricted
  from $X$ agrees with that from $\tau$ (i.e.
  $[X]|_{\Delta} = \tau|_{\Delta}$, or say of coherent
  orientation), then we assign
  $10J(\beta,\Delta) = 10j(\Delta,\beta|_{\Delta})$; otherwise,
  if $[X]|_{\Delta} = -\tau|_{\Delta}$ (or say decoherent
  orientation), then we assign
  $10J(\beta,\Delta) = \overline{10j}(\Delta,\beta_{\Delta})$.

  Now each $4$-simplex has a $10$j-symbol, which is
  just a linear map. Recall that the oriented $3$-simplices
  correspond to morphism spaces. We will contract the linear maps
  (taking a huge trace) using the fact that each $3$-simplex
  $\Delta'$ is the face of exactly two $4$-simplices. More
  concretely, observe that there are two cases.
  \begin{itemize}
    \item Both of them have coherent (or decoherent) orientations.
    \item One of them has coherent orientation, while the other
          has decoherent orientation.
  \end{itemize}
  In the first case, the corresponding vertices
  (\ref{graph/10j-symbol}) of the $C$-colored graphs that underly
  the assigned $10$j-symbols have the incoming and outgoing
  arrows exchanged. In the second case, the orientations of the
  arrows are the same but the colors are dual. Hence in both
  cases, we can contract the $10$j-symbols along $\Delta'$ as
  usual (\ref{def/contraction}). Since $X$ is a closed
  $4$-manifold, the final result is an element in the underlying
  field (i.e. a number).

  Finally, we define the Crane-Yetter state sum of $X$ to be the
  number
  $$\int_{X}^{CY} C \, := \, D^{2(n_{0}\text{-}n_{1})} \sum_{\beta} \prod_{f} dim(\beta(f)) \left(\ast \bigotimes_{\Delta} 10J(\beta;\Delta)\right),$$
  where $D^{2}$ denotes the global dimension of $C$, $n_{0}$
  denotes the amount of vertices, $n_{1}$ denotes the amount of
  edges, the sum runs over all possible $C$-colorings $\beta$ of
  $K$, the product runs through all faces $f$ of $K$ (recall
  $dim(x)=dim(x^{\star})$ for all $x$ in $C$), the tensor product
  runs through all $4$-simplices of $K$, and $\ast$ denotes the
  large contraction specified above.

  The result only depends on the PL-homeomorphism type of $X$ due
  to the invariance under Pachner moves. We refer the curious
  readers to the original paper
  \cite{crane-yetter/a-categorical-construction-of-4d-tqft}
  \cite{crane-kauffman-yetter/crane-yetter-state-sum}.
\end{definition}

\noindent The original state sum uses $15$j-symbols and therefore
involves a product running through the $3$-simplices. The term is
absent here because it is absorbed into the $10$j symbols. The
state sum is expected to be extended to a fully extended
topological quantum field theory (\cite[section 1.5]{bjs2018}
\cite{cooke2019excision}
\cite{integrating-quantum-groups-over-surfaces}
\cite{fac-homo--kirillov-tham}). For explicit evaluations of the
Crane-Yetter model see \cite{barenz/evaluation-crane-yetter} (for
numerical values on $4$-folds) and
\cite{guu/higher-genera-center} (for categorical values on
$2$-folds).

\subsection{Shadow state sum} \label{subsection/shadow-state-sum}

Throughout this subsection (\ref{subsection/shadow-state-sum}),
we fix an orientable shadowed $2$-polyhedron $P$ (over
$\mathbb{Z}[\frac{1}{2}]$, with boundary), a coordinated
premodular category $C$, and its set of simple objects $I$. Our
goal is to define the shadow state sum $\int_{P}^{sh} C$.

\begin{definition}[module of a trivalent graph]\label{def/module-of-a-trivalent-graph}
  Let $K_{0}$ be the empty graph and$\gamma$ be a trivalent
  graph. A $C$-coloring of $\gamma$ is a map
  $$\{\mbox{oriented edge of } \gamma\} \xrightarrow{\lambda} I, \quad \mbox{with } \lambda(e) = \lambda(-e)^{\star}.$$
  Define a $\mathbb{k}$-module
  $$H(\lambda) = \bigotimes_{x} H(\lambda_{x}, \lambda_{x}', \lambda_{x}''),$$
  where $H$ denotes the symmetrized modules
  (\ref{def/symmetrized-multiplicity-module}), $x$ runs through
  all vertices of $\gamma$ and the $\lambda_{x}$'s denote the
  colors assigned to the nearby edges oriented toward $x$.
  Define $\mathbb{k}$-modules
  $$H(\gamma) = \bigoplus_{\lambda \in color(\gamma; C)} H(\lambda), \quad H(K_{0}) = \mathbb{k},$$
  where $color(\gamma; C)$ denotes the set of $C$-colorings of
  $\gamma$.
\end{definition}

By a $C$-coloring of $P$ we mean a map $\phi$ from the set of
oriented regions of $P$ to $I$ such that
$\phi(\Sigma) = \phi(-\Sigma)^{\star}$. Denote by $color(P; C)$
the set of all $C$-colorings of $P$. An orientation of a $2$D
region induces an orientation on its edges by
$(\vec{n} \wedge \text{-})$, where $\vec{n}$ denotes a vector
pointing outward from the region. Therefore, a $C$-coloring
$\phi$ of $P$ induces a $C$-coloring $\partial \phi$ of its
boundary $\partial P$, a trivalent graph.

\begin{definition}[shadow state sum]\label{def/shadow-state-sum}\details{\cite[section X.1.2]{turaev-qiok-3-manifolds}}
  Every $C$-coloring on $\partial P$ extends to some $C$-coloring
  on $P$, so
  $H(\partial P) = \sum_{\phi \in color(P; C)} H(\partial \phi)$.
  Fix a $\phi \in color(P; C)$, and define the following
  $\mathbb{k}$-modules and vectors.
\begin{itemize}
  \item For each oriented edge $\vec{e}$ in
        $P \setminus \partial P$, define $H_{\phi}(\vec{e})$ to
        be $H(i,i',i'')$ where the $i$'s are the colors assigned
        to the three adjacent regions compatibly oriented with
        $\vec{e}$.
  \item For each (unoriented) edge $e$ in
        $P \setminus \partial P$, define $H_{\phi}(e)$ to be the
        non-ordered tensor product
        $H_{\phi}(\vec{e}) \otimes H_{\phi}(\text{-}\vec{e})$
        with an arbitrary orientation $\vec{e}$. The pairing
        (\ref{def/contraction}) defines a canonical vector
        $|e|_{\phi} \in H_{\phi}(e)$.
  \item For each tetrahedral point $x \in P$, pick a small enough
        neighborhood $U$ of $x$ in $P$ homeomorphic to the cone
        of the $1$-skeleton of the boundary of some tetrahedron.
        The closure $\overline{U}$ a $C$-colored $2$-polyhedron
        with four boundary line points
        $x_{0}, x_{1}, x_{2}, x_{3}$ and six $C$-colored regions.
        Denote by $\phi_{ij}$ the color for the oriented region
        $\overrightarrow{xx_{i}x_{j}}$ (clearly,
        $\phi_{ij} = \phi_{ji}^{\star}$). Finally, define a
        vector and a $\mathbb{k}$-module
        $$|x|_{\phi} := \normalizedSixJSymbol{\phi_{01}}{\phi_{02}}{\phi_{30}}{\phi_{32}}{\phi_{13}}{\phi_{21}} \in \bigotimes_{i=0}^{3} H_{\phi}(\overrightarrow{x_{i}x}) =: H_{\phi}(x),$$
        where $\otimes$ denotes the unordered tensor product of
        $\mathbb{k}$-modules. The result is independent to the
        labeling $0, 1, 2, 3$.
\end{itemize}

\noindent The procedure above defines a vector in the
$\mathbb{k}$-module
$$(\otimes_{x} |x|_{\phi}) \otimes (\otimes_{e} |e|_{\phi}) \in \left( \bigotimes_{x} H_{\phi}(x) \right) \otimes \left( \bigotimes_{e} H_{\phi}(e) \right)$$
where $x$ runs over all tetrahedral points of $P$, and $e$ runs
over all (nonoriented) edges of $P \setminus \partial P$. By
contracting the vector along all tetrahedral points $x$ and all
edges $e$ whose boundary points are not both in $\partial P$, we
obtain a vector in $|\phi| \in H(\partial \phi)$. Finally, we
define the shadow state sum to be
$$\left(\int^{sh}_{P} C \right) = \left( D^{\text{-} b_{2}(P) \text{-} null(P)} \sum_{\phi \in color(P)} \sigma_{\phi} \, |\phi| \right) \in \sum_{\phi} H(\partial \phi) = H(\partial P),$$
where $D$ denotes the global dimension, $b_{2}$ denotes the
second betti number, $null$ denotes the nullity
(\ref{def/nullity-of-a-shadowed-2-polyhedron}), and
$\sigma_{\phi} \in \mathbb{k}$ is a normalizing constant defined
as
$$\sigma_{\phi} = \prod_{e} dim_{C}'(\partial\phi(e))^{-1} \prod_{Y} dim_{C}(\phi(Y))^{\chi(Y)} \nu'_{\phi}(Y)^{2gl(Y)} \prod_{g} dim_{\mathbb{k}}(Hom_{C}(V_{0}, V_{i} \otimes V_{j} \otimes V_{k})),$$
where $e$ runs over edges of $\partial P$ (but not circle
$1$-strata), $Y$ runs over regions of $X$, $g$ runs over circle
$1$-strata of $sing(X)$, and $\chi$ denotes the Euler
characteristics.
\end{definition}

\begin{proposition}[shadow state sum is invariant under stable
  shadow
  move]\label{prop/shadow-state-sum-is-invariant-under-stable-shadow-move}
  Let $C$ be a premodular category and $P, P'$ be $2$-polyhedra
  that are equal as stable shadows. Then
  $$\int^{sh}_{P} C = \int^{sh}_{P'} C.$$
  Namely, shadow state sum is invariant under stable shadow move.
\end{proposition}
\begin{proof}
  We start with the special case where $C$ is a modular category.
  For invariance under basic shadow moves, the essential
  ingredients are the orthonormality relation, the Racah
  identity, and the Biedenharn-Elliott identity
  (\ref{prop/basic-equalities-of-6j-symbols}); see \cite[theorem
  X.2.1]{turaev-qiok-3-manifolds} for a proof. For invariance
  under addition of $S^{2}_{0}$, it boils down to proving the
  addition formula $$|P_{1} + P_{2}| = |P_{1}| \otimes |P_{2}|$$
  \cite[theorem X.2.2]{turaev-qiok-3-manifolds} and using the
  equality $|S^{2}_{0}| = D^{\text{-}2} \sum_{i \in I} dim(i)^{2} = 1$.
  Both proofs carry through verbatim to the premodular case.
\end{proof}

% TODO (less urgent): Provide explicit examples. This is
% important for expositing.

\begin{definition}[shadow state sum of a $4$-manifold]\label{def/shadow-state-sum-of-a-4-manifold}
  Let $X$ be a closed $4$-manifold, $C$ a coordinated premodular
  category, $stab([P])$ a stable shadow of $X$ represented by a
  shadowed $2$-polyhedron $P$. Define the shadow state sum
  $\int_{X}^{sh} C$ of $X$ to be $\int_{P}^{sh} C$, which is
  well-defined by \ref{remark/stable-shadow-of-a-4-manifold} and
  \ref{prop/shadow-state-sum-is-invariant-under-stable-shadow-move}.
\end{definition}

\subsection{Main result: equivalence of state sums}

\noindent We state and prove the main theorem of this paper. See
\ref{subsection/a-summary-to-experts} for the main idea of the
proof.

\begin{theorem}[equivalence of state sums]\label{theorem/main-theorem}
  Let $X$ be a closed $4$-manifold and $C$ be a coordinated
  premodular category. Then
  $$\int^{CY}_{X}C = \int^{sh}_{X} C.$$
  Namely, their Crane-Yetter state sum and shadow state sum are
  equal.
\end{theorem}

\begin{remark}
  The Witten-Reshetikhin-Turaev (quantum Chern-Simons) model is
  known to be the boundary theory of Crane-Yetter model
  \cite{barrett/observables-in-tv-and-cy} \cite{tham/phd-thesis}.
  It is also shown that the former is the boundary theory of the
  shadow TQFT \cite[X.3.2 \& theorem
  X.3.3]{turaev-qiok-3-manifolds}. Therefore, theorem
  \ref{theorem/main-theorem} provides another proof for the first
  fact.
\end{remark}

\begin{proof}
  Fix a small $\epsilon > 0$.

  We begin by computing the shadow state sum for $X$. First, fix
  a triangulation $T$ for $X$. Denote by $T_{i}$ to be the set of
  $i$-cells of $T$, and fix a total ordering on $T_{0}$. Recall
  that the ordering induces an orientation of each $4$-cell. If
  it agrees with the orientation from $X$, we call it an
  coherently oriented cell; otherwise a decoherently oriented
  cell. From $T$ we will construct a shadow of a similar
  ``shape''. Indeed, the dual of the any triangulation provides a
  handle decomposition $H$, in which each $k$-cell corresponds to
  an $(n-k)$-handle. The construction in
  (\ref{def/shadow-of-a-4-manifold-from-a-handle-decomposition})
  constructs a shadow for $X$, as follows.

  Take the union of the $0$-handles and the $1$-handles. Its
  boundary $Y$ is a connected sum of $|T_{0}|$ $3$-spheres. By
  (\ref{def/shadow-of-a-4-manifold-from-a-handle-decomposition}
  and \ref{def/shadow-cone-of-a-framed-link-in-a-3-manifold}), we
  need to pick a skeleton of $Y$. We will construct a very
  concrete one as follows. Within each $4$-simplex $\Delta$,
  $Y \cap \Delta$ is $S^{3} \setminus (5 \times B^{3})$. Think of
  this as
  $\mathbb{R}^{3} \setminus \bigcup_{\vec{v}} B_{\epsilon}(\vec{v})$,
  where $B_{r}(x)$ denotes the ball of radius $r$ centered at
  $x$, and $\vec{v}$ runs through the set
  $\{(1,0,0), (0,1,0), (\text{-}1,0,0), (0,\text{-}1,0)\}$.
  Denote $S^{1}(\vec{v})$ to be the equator dual to $\vec{v}$ of
  each $2$-sphere
  $S^{2}(\vec{v}) := \partial B_{\epsilon}(\vec{v})$. The largest
  component of
  $B_{0}(1) \setminus \left(\bigcup_{\vec{v}}S^{1}(\vec{v})\right)$
  is a $4$-punctured $2$-sphere. Finally, remove an
  $\epsilon$-disk centered at $(0,0,1)$ from it, and let the
  boundary straightly stretch to $(0,0,+\infty)$; this is a
  $5$-punctured $2$-sphere $\Sigma$, which is a local skeleton
  for $Y$.

  To continue following
  (\ref{def/shadow-of-a-4-manifold-from-a-handle-decomposition}),
  we need to project to the links (the gluing data of the
  $2$-handles in $Y$) to $\Sigma$. It is a geometric exercise to
  construct a projection so that the projected diagrams look as
  follows

  \begin{center}
    \includegraphics[height=6cm]{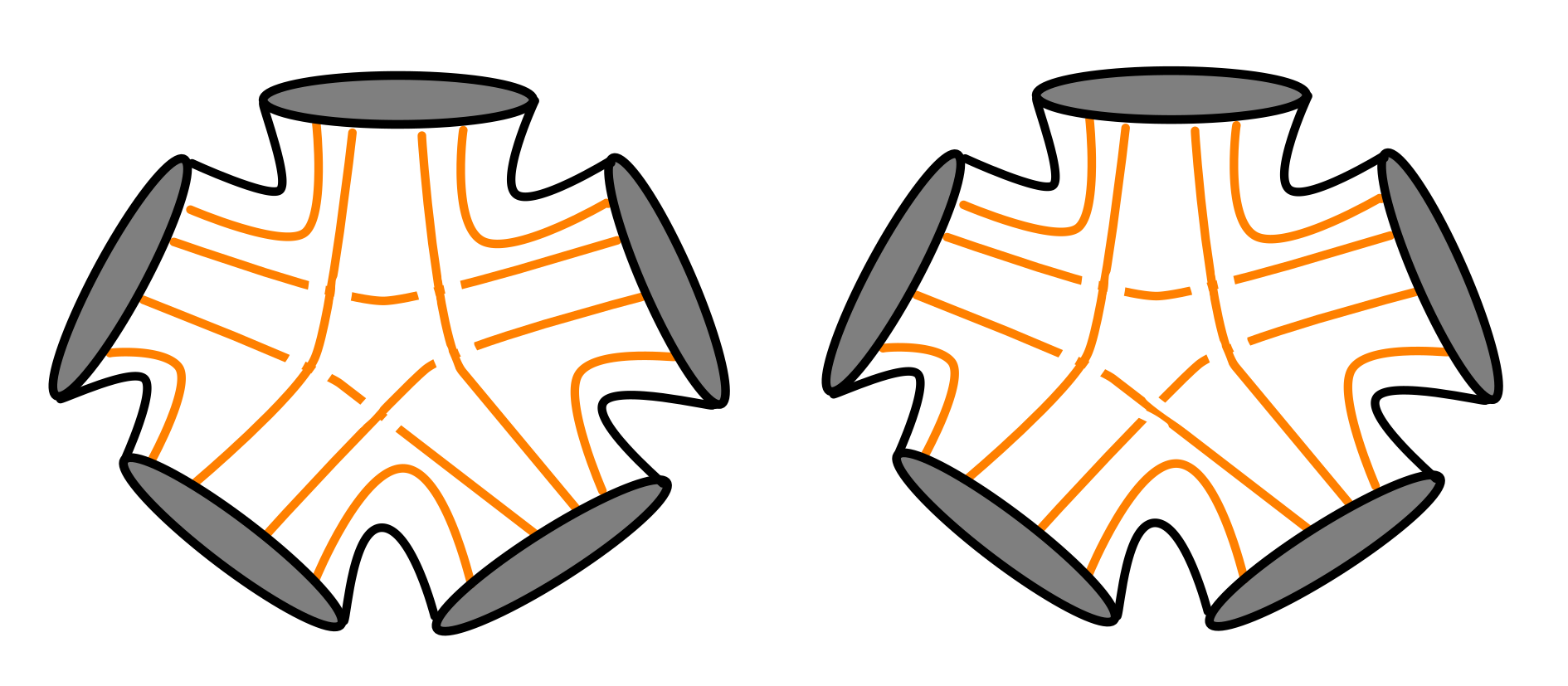}
  \end{center}

  \noindent depending on whether the $4$-cell is coherently
  oriented or not. The construction then cones the projected
  links, and assigns gleams around each intersection of links on
  $\Sigma$. This encodes a local piece of the complete shadow,
  which is a gleamed $2$-polyhedron $P$ without boundary.
  However, $P$ also looks the same locally within each
  $4$-simplex (up to mirror), so for simplicity we will keep
  working locally.

  The shadow state sum, locally, is represented by the following
  diagram and its mirrored image.

  \begin{center}
    \includegraphics[height=11cm]{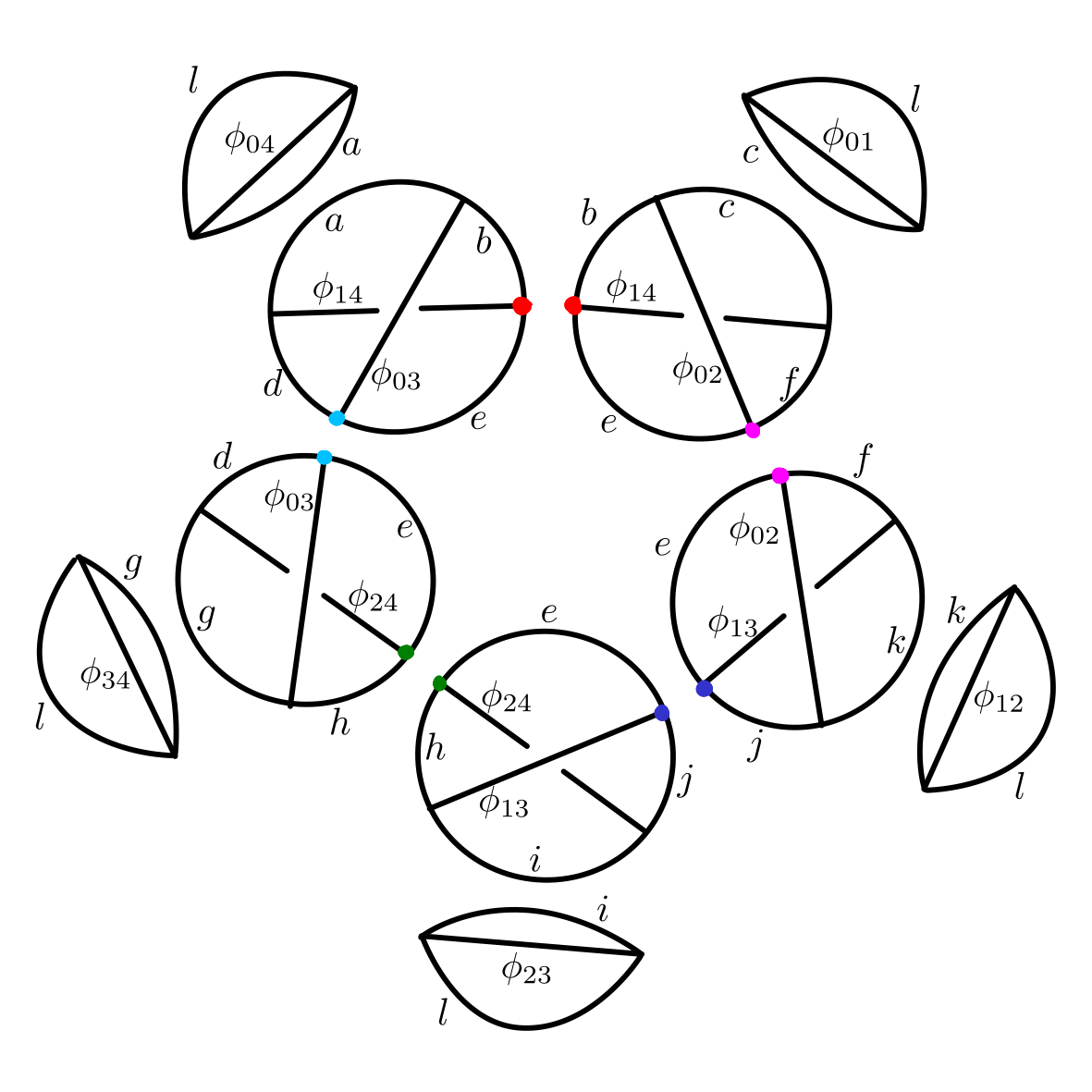}
  \end{center}

  \noindent where each of the $5$ tetrahedral graphs is obtained
  by the $6$j-symbol twisted with the $4$ gleams around the
  corresponding tetrahedral point. The vertices in the graph are
  paired (indicated by the colors), and paired vertices are
  actually labeled by elements of the bases and dual bases of the
  morphism spaces. We contract them and obtain the following
  diagram using the techniques given in
  \cite[Lemma 1.1, 1.3]{kirillov-balsam/turaev-viro-I}.

  \begin{center}
    \includegraphics[height=11cm]{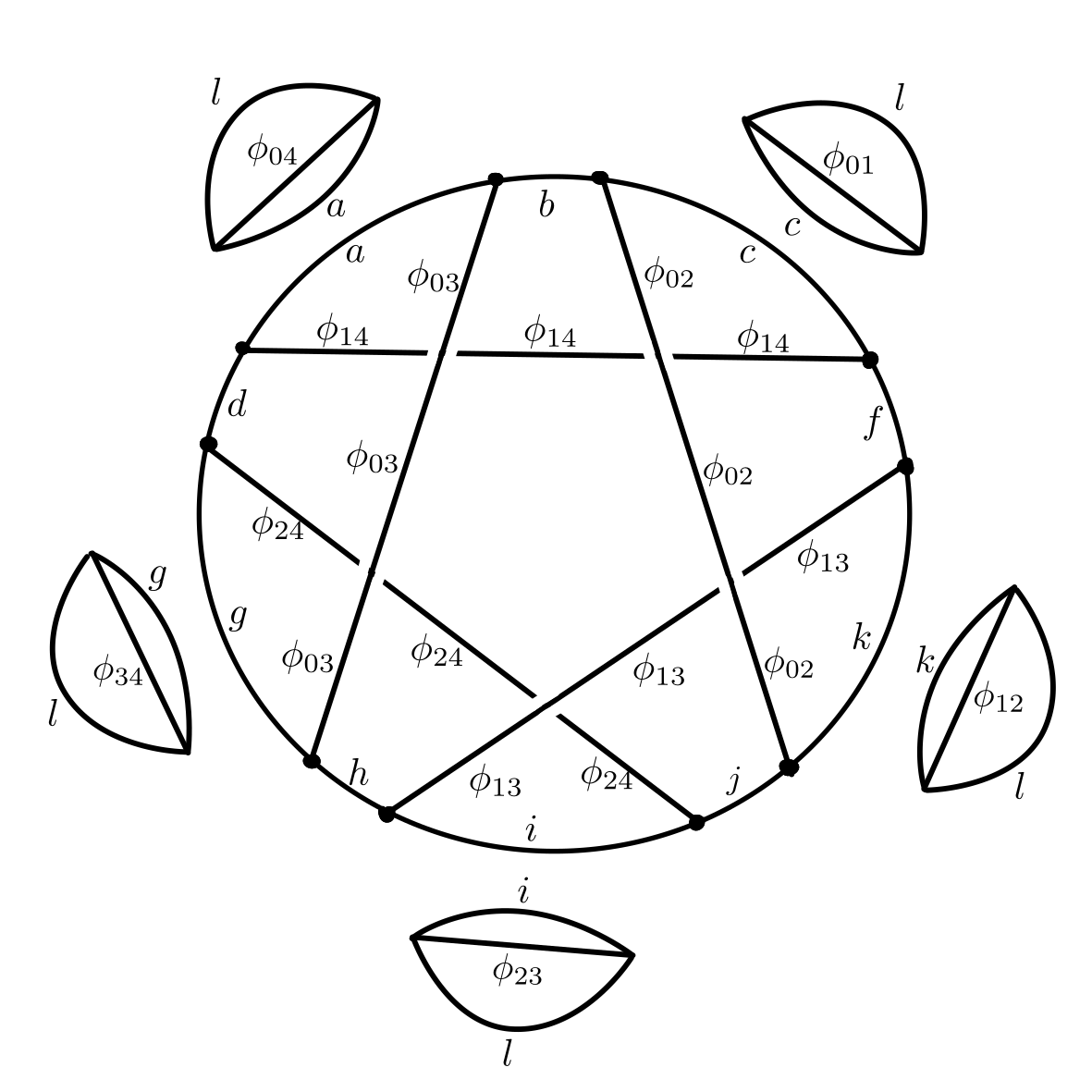}
  \end{center}

  \noindent % Using 2 lemmas (cite TODO??), ..
  We can further contract each theta graph to the central
  component using the following procedure.

  \begin{center}
    \includegraphics[height=18cm]{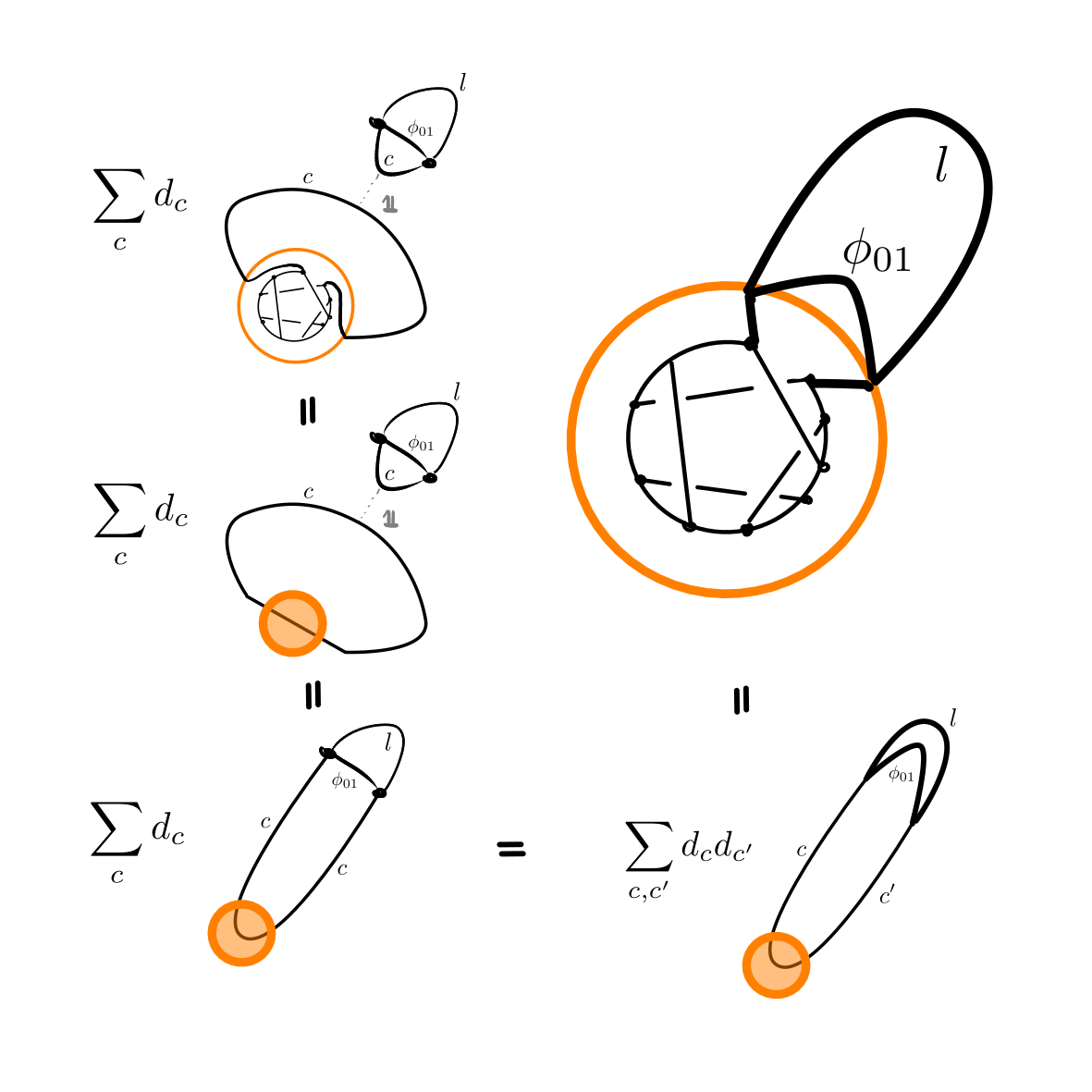}
  \end{center}
  \noindent Repeat for five times, and the result is the
  following

  \begin{center}
    \includegraphics[height=11cm]{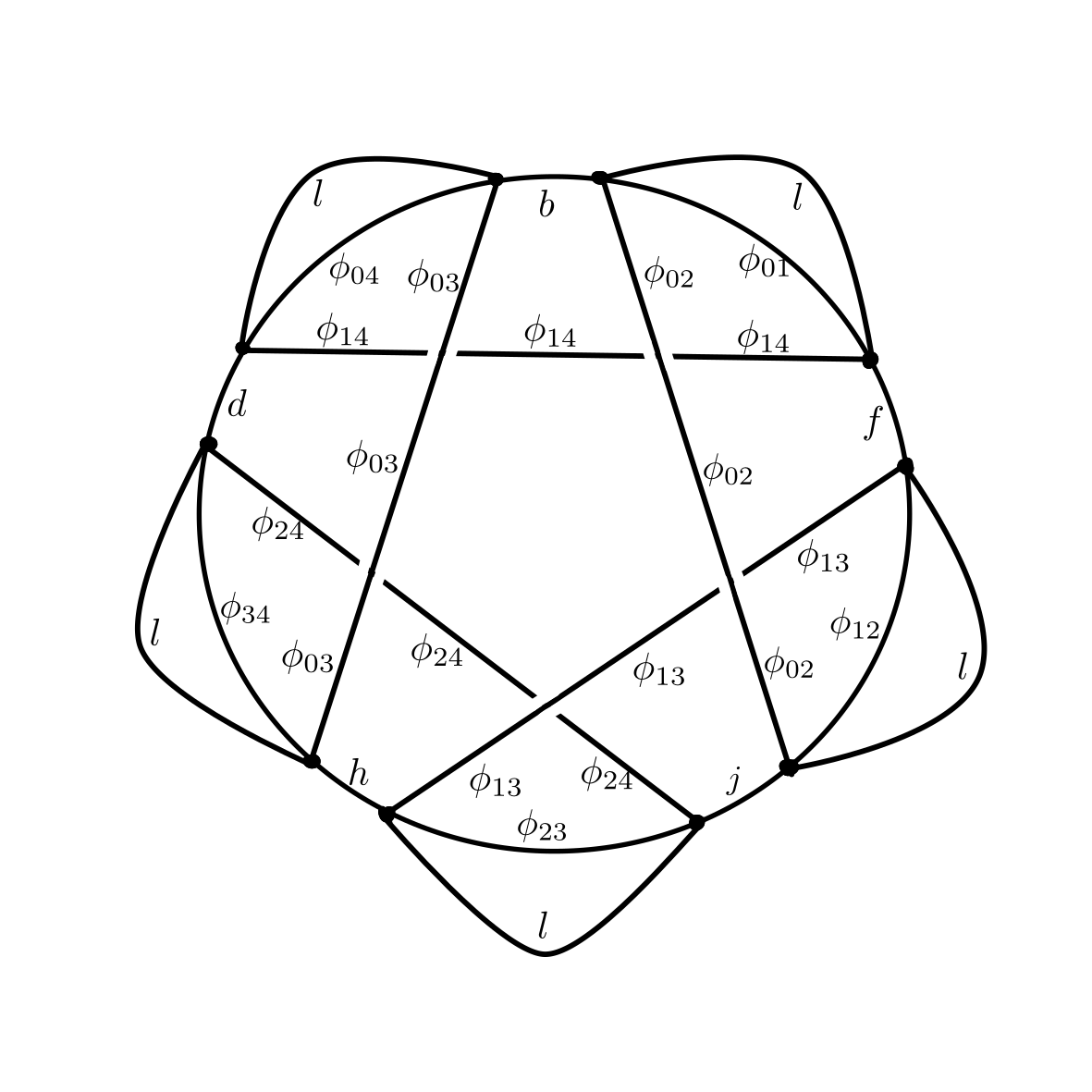}
  \end{center}

  \noindent This is almost the $10$j-symbol involved in the
  definition of the $\int_{X}^{CY} C$, except the extra edges
  labeled by $b, d, f, h ,j, l$. However, after contracting the
  local pieces together, the extra edges form unlinks (colored by
  the regular coloring $\Omega = \sum_{i\in I}dim(i)i$) and
  therefore can be viewed as a factor $D^{2}$ and removed from
  the diagram. The rest of the proof is by counting.
\end{proof}

\printbibliography
\end{document}